\documentclass[12pt]{article}

\usepackage{latexsym,amsmath,amscd,amssymb,amsthm,graphics}
\usepackage{enumerate}
\usepackage[margin=2cm]{geometry}

\usepackage{graphicx}

\usepackage[square,authoryear]{natbib}
\usepackage[colorlinks]{hyperref}
\usepackage{url}
\usepackage{framed}

\usepackage[all]{xy}

\makeatletter

\@addtoreset{figure}{section}
\def\thefigure{\thesection.\@arabic\c@figure}
\def\fps@figure{h, t}
\@addtoreset{table}{bsection}
\def\thetable{\thesection.\@arabic\c@table}
\def\fps@table{h, t}
\@addtoreset{equation}{section}

\makeatother

\begin{document}

\newtheorem{theorem}{Theorem}[section]
\newtheorem{definition}[theorem]{Definition}
\newtheorem{lemma}[theorem]{Lemma}
\newtheorem{remark}[theorem]{Remark}
\newtheorem{proposition}[theorem]{Proposition}
\newtheorem{corollary}[theorem]{Corollary}
\newtheorem{example}[theorem]{Example}

\def\balpha{\boldsymbol\alpha}
\def\bom{\boldsymbol\omega}
\def\bOm{\boldsymbol\Omega}
\def\bPi{\boldsymbol\Pi}
\def\below#1#2{\mathrel{\mathop{#1}\limits_{#2}}}



\title{Selective decay by Casimir dissipation in inviscid fluids}

\author{
Fran\c{c}ois Gay-Balmaz$^{1}$ and Darryl D. Holm$^{2}$
}
\addtocounter{footnote}{1}
\footnotetext{Laboratoire de M\'et\'eorologie Dynamique, \'Ecole Normale Sup\'erieure/CNRS, Paris, France. 
\texttt{gaybalma@lmd.ens.fr}
\addtocounter{footnote}{1} }
\footnotetext{Department of Mathematics, Imperial College, London SW7 2AZ, UK. 
\texttt{d.holm@ic.ac.uk}
\addtocounter{footnote}{1}}

\date{
\\Hamiltonian and Lagrangian mechanics, 45.20.Jj, 47.10.Df; 
\\Fluids, mathematical formulations, 47.10.A- }

\maketitle

\makeatother

\begin{abstract}
The problem of parameterizing the interactions of larger scales and smaller scales in fluid flows is addressed by considering a property of two-dimensional incompressible turbulence. The property we consider is selective decay, in which a Casimir of the ideal formulation (enstrophy in 2D flows, helicity in 3D flows) decays in time, while the energy stays essentially constant. 
This paper introduces a mechanism that produces selective decay by enforcing Casimir dissipation in fluid dynamics. This mechanism turns out to be related in certain cases to the numerical method of anticipated vorticity discussed in \cite{SaBa1981,SaBa1985}. 
Several examples are given and a general theory of selective decay is developed that uses the Lie-Poisson structure of the ideal theory. A scale-selection operator allows the resulting modifications of the fluid motion equations to be interpreted in several examples as parameterizing the nonlinear, dynamical interactions between disparate scales. The type of modified fluid equation systems derived here may be useful in modelling  turbulent geophysical flows where it is computationally prohibitive to rely on the slower, indirect effects of a realistic viscosity, such as in large-scale, coherent, oceanic flows interacting with much smaller eddies.
\end{abstract}

\maketitle

\newpage

\tableofcontents



\section{Introduction}

\subsection{Symmetries and selective decay}

From the viewpoint of Noether's theorem, energy is conserved in ideal fluid dynamics because of the time-translation symmetry of the Lagrangian in Hamilton's principle for ideal fluid motion. A second type of fluid conservation law arises via Noether's theorem because of relabelling symmetry of the Lagrangian. Relabelling symmetries smoothly transform the labels of the fluid parcels without changing the Eulerian quantities on which Hamilton's principle for ideal fluids depends. 
The conservation laws associated with \emph{relabelling symmetries} are called Casimirs, because in the Hamiltonian formulation of ideal fluid dynamics in the Eulerian representation their Lie-Poisson brackets with \emph{any} other functionals vanish identically.
Thus, the Casimirs arise from a continuous symmetry of the Eulerian representation of ideal fluid dynamics. This relabelling symmetry is also responsible for Kelvin's circulation theorem in ideal fluid dynamics, which immediately leads to the conservation of the Casimirs for 2D ideal incompressible flow, see \cite{HMR1998,CoHo2012}. The Kelvin circulation theorem and the associated Casimir conservation laws are \emph{kinematic}, because they hold for any choice of 
Hamiltonian in the Eulerian representation. 
Energy conservation is \emph{dynamic}. 

The two types of conservation laws for ideal fluids (energy and Casimirs) have been used together previously, in particular, as the basis for the energy-Casimir method of nonlinear stability analysis of equilibrium flows \cite{HMRW1985}. This method is based on the concept of \emph{dynamical balance} in ideal fluids. In particular, the equilibrium solutions of ideal fluid dynamics are characterized by such a balance.  Namely, the critical points $\delta(h+\lambda C)=0$ of any linear combination $h+\lambda C$ of the energy $h$, a Casimir $C$ and constant $\lambda$ occur for 
equilibrium solutions of Euler's equations for ideal fluids. 
An initial disturbance of this balance may be either stable or unstable in the Lyapunov sense. Sufficient conditions for linear Lyapunov stability of an equilibrium flow may be determined by taking the second variation of the energy-Casimir sum, $\delta^2(h+\lambda C)|_e$, where subscript $e$ means evaluated at the equilibrium flow. This second variation is \emph{conserved} by the linearized equations. This means that linear Lyapunov stability is implied for those flow equilibria for which the second variation $\delta^2(h+\lambda C)|_e$ comprises a norm in the space of perturbations. Thus, in the ideal Hamiltonian description of fluids, the equilibrium solutions and their Lyapunov stability are determined by the interplay of energy and Casimirs.
For more details, and discussions of the extension of the energy-Casimir method to \emph{nonlinear} Lyapunov stability analysis, see  \cite{HMRW1985} and references therein, particularly \cite{Ar1965a,Ar1965b} for the case of 2D ideal incompressible fluid flows. 

\paragraph{The selective decay hypothesis in 2D turbulence.}
Already in \cite{Fj1953} it was shown that, in contrast to 3D incompressible flow, the existence of two constants of motion (energy and Casimir) implies that any transfer of kinetic energy to higher wavenumbers $(k)$ in the 2D incompressible flow of a viscous fluid must be accompanied by a larger transfer of energy to lower wavenumbers. Over time, this inverse cascade property leads to \emph{coarse-graining}. That is, it leads to the emergence of large-scale features of the flow. An explanation for this effect based on statistical mechanics and equilibrium thermodynamics was offered in \cite{Kr1967,Kr1971}, which predicted  a forward cascade $(k^{-3})$ of enstrophy and an inverse cascade $(k^{-5/3})$ of kinetic energy in 2D incompressible turbulence, as reviewed in \cite{KrMo1980}.  Investigating  these predictions numerically, \cite{MaMo1980} performed simulations of 2D incompressible turbulence using the Navier-Stokes equations and on this basis postulated the \emph{selective decay} hypothesis. Their selective decay hypothesis was that in the 2D incompressible flow of a viscous fluid the squared $L^2$ norm of vorticity (or enstrophy, the quadratic Casimir) should decay in time, while the squared $L^2$ norm of velocity (kinetic energy) which has contributions from both directions in wavenumber should stay essentially constant.
The balance between the $(k^{-5/3})$ and $(k^{-3})$ cascades occurs at the length scale at which  energy flux passes through zero, so that energy may be passed in both directions at the same rate simultaneously. 

Because the balance induced by the energy-Casimir dynamics is so important in the motion of fluids, one could imagine designing \emph{modified fluid equations} that would govern the balance by \emph{imposing} the selective decay of one or the other of these two types of conservation laws. 
Moreover, because this balance in the ideal case stems from the Lie-Poisson bracket in the Hamiltonian formulation of ideal fluid dynamics, one could imagine that an effective approach in deriving such modified equations would involve using elements of the Lie-Poisson bracket that would select which type of conservation law would dominate and thus govern the properties of the solution. 

The two types of ideal fluid constants of motion, energy and Casimirs, typically have quite different dependences on spatial gradients of the solutions. Consequently, the interplay between them can be interpreted as an interaction between larger and smaller scales of excitation, as specified by coherence lengths, or spectral wavenumbers, etc. Thus, imposition of a selective decay mechanism for one or the other type could in principle control the dominance of one set of scales, or sizes of excitations, over another. 
Hence, modifying the fluid equations to  impose selective decay of either the energy or a particular Casimir is interesting, because such a modification could control the direction of the energy cascade, i.e., whether it is forward (toward smaller, finer scales), or backward (toward larger, coarser scales).
This means that modifications of the fluid equations that promote selective decay could control the dominant sizes that emerge in the flow and support their stability, for example, by causing smaller perturbations to decay away.
Such modifications would control the rate at which the cascade occurs, in either direction, and thus would control the type of solution that emerges from a given initial condition. This sort of control could conceivably even be used to regularise the solution by choosing to diminish smaller-scale features and thereby promote the emergence of larger-scale features of the flow. 

Various types of modifications of the Poisson bracket for Hamiltonian systems have previously been proposed that produce energy dissipation. These proposed modifications include: (i) adding a symmetric bilinear form, \cite{Ka1984,Mo1984,Gr1984,Ot2005}), and (ii) application of a double bracket, \cite{Br1991,BKMR1996,HoPuTr2008,BrElHo2008}. 
See also \cite{VaCaYo1989}, in which a modification of the transport velocity was used to impose energy dissipation with fixed Casimirs. 

Here we are not considering energy dissipation. Instead, we are considering a type of selective decay of ideal Casimir constants of motion at fixed energy, as was discussed for 2D turbulence cascades in \cite{MaMo1980}. In the present case, the selective decay is enforced by a modification of the vorticity equation that is based on the Lie-Poisson structure of its Hamiltonian formulation in the ideal case  \cite{Ar1966,Ar1969,Ar1978,HMR1998}. 

Thus, the aim of this paper is to impose selective decay by using the Lie-Poisson structure of the ideal theory, and interpret the resulting modifications of the equations as a means of dynamically and nonlinearly parameterizing the interactions between disparate scales directly. This type of modification may be useful in situations where it is computationally prohibitive to rely on the slower, indirect effects of viscosity.  
Moreover, these ideas should apply not only to 2D flow, but also to any 3D flow that is sufficiently anisotropic, particularly in geophysical fluid flows in thin domains.

\paragraph{Plan of the paper.}

The remainder of this introductory section provides more background about the two types of conservation laws in fluids and defines our selective decay mechanism based on Casimir dissipation in the example of 2D incompressible flows. 
Remarkably, in that 2D case,
the mechanism for selective decay by Casimir dissipation in 2D incompressible fluid dynamics turns out to be related to the numerical method of \emph{anticipated vorticity} in \cite{SaBa1981,SaBa1985}. Section \ref{GenTheory-sec} discusses the general theory of selective decay by Casimir dissipation in the Lie algebraic context that underlies the Lie-Poisson Hamiltonian formulation of ideal fluid dynamics, as explained in, e.g., \cite{HMR1998}. In particular, we develop the Kelvin circulation theorem and Lagrange-d'Alembert variational principle for Casimir dissipation. 
In the Lagrange-d'Alembert formulation, the modification of the motion equation to impose selective decay is seen as an energy-conserving constraint force, something like the Lorentz force.
Section \ref{semidirect-sec} extends the Casimir dissipation theory to include fluids that possess advected quantities such as heat, mass, buoyancy, magnetic field, etc., by using the standard method of Lie-Poisson brackets for semidirect-product actions of Lie groups on vector spaces reviewed in \cite{HMRW1985}. Here, the main examples are the rotating shallow water equations  and the 3D Boussinesq equations for rotating stratified incompressible fluid flows. Finally, section \ref{conclusion-sec} summarizes our conclusions and suggests some possible next directions for applications of the selective decay modifications of fluid equations, treated as dynamical parameterizations of the interactions between disparate scales.

\subsection{2D flows}
An interesting feature of ideal incompressible 2D fluid flows is that they have two types of conservation laws which arise from their Hamiltonian formulation in terms of a Lie-Poisson bracket, $\{\,\cdot\,,\,\cdot\,\}$, \cite{Ar1966,Ar1969,Ar1978}, as
\begin{equation}
\frac{df( \omega )}{dt}
=\{f,h\}( \omega )
= \left\langle \omega , \left[ \frac{\delta f}{\delta \omega }, \frac{\delta h}{\delta \omega }\right] \right\rangle  
:= \int_\mathcal{D} \omega \left[ \frac{\delta f}{\delta \omega }, \frac{\delta h}{\delta \omega }\right] \,dxdy\,,
\label{ArnoldLPB}
\end{equation}
where $\omega$ is the (scalar) vorticity of the 2D flow, the square bracket $[\,\cdot\,,\,\cdot\,]$ is the 2D Jacobian, written in various forms as 
\begin{equation}
[f,h]=J(f,h)=f_xh_y-h_xf_y
=\mathbf{\hat{z}}\cdot\nabla f\times\nabla h
=\mathbf{\hat{z}}\cdot 
{\rm curl}\,(f\mathbf{\hat{z}}) \times {\rm curl}\,(h\mathbf{\hat{z}})
\,,\label{brkt-idnty}
\end{equation}
and the angle bracket $\langle\,\cdot\,,\,\cdot\,\rangle$ in (\ref{ArnoldLPB}) is the $L^2$ pairing in the domain $\mathcal{D}$ of the $(x,y)$ plane. For convenience, we shall take the domain $\mathcal{D}$ to be periodic, so we need not worry about boundary terms arising from integrations by parts.

\paragraph{Energy.} The first type of conservation law for such 2D ideal fluid flows arises from the antisymmetry of their Lie-Poisson bracket. Namely, antisymmetry of the bracket implies conservation of the Hamiltonian function $h$
\[
\frac{dh( \omega )}{dt}
=\{h,h\}( \omega )
= 0
\,,
\]
for any given choice of $h$.

\paragraph{Casimirs.} The second type of conservation law arises because their Lie-Poisson bracket has a kernel (i.e., is degenerate), which means there exist functions $C( \omega )$ for which 
\begin{equation}
\frac{dC( \omega )}{dt}
=\{C,h\}( \omega ) = 0
\,,
\label{Casimir-def}
\end{equation}
for any Hamiltonian $h(\omega)$.
Functions that satisfy this relation for any Hamiltonian are called \emph{Casimir functions}. (Lie called them distinguished functions, according to \cite{Ol2000}.)

\begin{theorem}[\cite{Ar1966,Ar1969,Ar1978}]\rm
The Casimirs for the Lie--Poisson bracket (\ref{ArnoldLPB}) in the Hamiltonian formulation of 2D incompressible ideal fluid motion are
\[
C_\Phi = \int_\mathcal{D} \Phi(\omega)\,dxdy
\,,\]
for any smooth function $\Phi$.
\end{theorem}

\begin{proof}
The well-known proof of this statement follows from an identity for the $(x,y)$ Jacobian bracket $[\,\cdot\,,\,\cdot\,]$ that arises from integration by parts as
\begin{equation}
\int_\mathcal{D} a[b,c]\,dxdy
= \int_\mathcal{D} b[c,a]\,dxdy
= \int_\mathcal{D} c[a,b]\,dxdy
\,.
\label{2D-integparts-id}
\end{equation}
This identity holds for any smooth functions $a,b,c\in\mathcal{F}(\mathcal{D})$,
provided the boundary terms of the domain $\mathcal{D}$ do not contribute; for example, when the domain $\mathcal{D}$ is taken to be periodic. The boundary terms arising in integration by parts will be discussed further in a moment.
Thus, 
\[
\{C_\Phi,h\}( \omega )
= \int_\mathcal{D} 
\omega \left[ {\Phi}' (\omega), \frac{\delta h}{\delta \omega }\right] \,dxdy
= \int_\mathcal{D} 
\frac{\delta h}{\delta \omega } \left[\omega, {\Phi}' (\omega) \right] \,dxdy
=
0
\,,
\]
for any Hamiltonian $h$ and function $\Phi(\omega)$. 
\end{proof}

\paragraph{Enstrophy.} Among the Casimirs $C_\Phi(\omega)$ for 2D ideal fluids is the famous \emph{enstrophy}, given by
\begin{equation}
C_2 := 
\frac12 \|\omega\|_{L^2(\mathcal{D})}^2 = \frac12 \int_\mathcal{D} \omega^2\,dxdy
\,,
\label{enstrophy-def}
\end{equation}
which is the Casimir of the Lie-Poisson bracket (\ref{ArnoldLPB}) for the case that $\Phi(\omega)=\frac12\omega^2$. 

\begin{remark}\rm
Any smooth function of a Casimir (such as its square) is also a Casimir. 
\end{remark}

\paragraph{Turbulence and selective decay.}
When viscosity is added, the Navier-Stokes equations for 2D incompressible turbulence result from the equations above. These equations are  dissipative and, in the absence of forcing, both types of conserved quantities in the ideal case would decay in time and eventually vanish. 
Remarkably, in 2D incompressible turbulence these two types of conservation laws are found numerically to decay at \emph{different rates} and this difference is found to have an important emergent effect on the spectral properties and statistics of 2D incompressible turbulence, \cite{Kr1967,Kr1971}. In particular, this feature of \emph{selective decay} in 2D turbulence leads to an inverse cascade of energy toward larger scales and a forward cascade of enstrophy toward smaller scales where its dissipation occurs, \cite{MaMo1980}. This emergent effect of viscosity in 2D incompressible turbulence has had a long history of investigation and is thought to be important in many applications, including climate modelling. 
Recent work tends to understand this effect as a mechanism for parameterizing the interactions between disparate scales, for example, between large coherent oceanic flows and much smaller eddies, \cite{MaAd2010}. For other recent discussions of these ideas in modelling disparate scale interactions in geophysical and astrophysical turbulence, see \cite{MiPoSu2008}.

\paragraph{Question.} Given that the two types of conservation laws in ideal 2D flow are both properties of the Lie-Poisson bracket for such flows, and their dissipation has a profound effect of the properties of the flow, it is natural to ask whether one may use the Lie-Poisson bracket of ideal 2D flow to \emph{impose} a process of selective decay. 

\subsection{Casimir dissipation}\label{Casimir_dissipation} 

The Lie-Poisson bracket may be used to introduce a type of ``nonlinear viscosity'' that preserves the energy, but dissipates a given Casimir. This may be accomplished naturally by modifying the Hamiltonian formulation to introduce a \emph{quadratic} Lie-Poisson bracket structure, as follows,
\begin{equation}
\frac{df( \omega )}{dt}
= \int_\mathcal{D}\omega  \left[ \frac{\delta f}{\delta \omega } , \frac{\delta h}{\delta \omega } \right] 
\,dxdy
- \theta \!\! \int_\mathcal{D}\left[ \frac{\delta f}{\delta \omega } , \frac{\delta h}{\delta \omega } \right]
L
\left[ \frac{\delta C}{\delta \omega }, \frac{\delta h}{\delta \omega }\right]\,dxdy
\,,
\label{2LPB-2D}
\end{equation}
where $\theta$ is a given constant and $L$ is an arbitrary positive self-adjoint linear differential operator. The differential operator $L$ allows a degree of scale selection in the Casimir dissipation approach. For example, 
choosing $L=(1-\alpha^2\Delta)^s$ would define a Sobolev $H^s( \mathcal{D} )$ inner product with length scale $\alpha$. 
One immediately sees from equation \eqref{2LPB-2D} that 
\[
\frac{dh( \omega )}{dt} = 0
\quad\hbox{and}\quad
\frac{dC( \omega )}{dt} = 
- \,\theta \!\! 
\int_\mathcal{D} \left[ \frac{\delta C}{\delta \omega } , \frac{\delta h}{\delta \omega } \right]
L
\left[ \frac{\delta C}{\delta \omega } , \frac{\delta h}{\delta \omega } \right]
dxdy
=: - \,\theta
\left\|\left[ \frac{\delta C}{\delta \omega }, \frac{\delta h}{\delta \omega }\right]\right\|^2_L
\,.\]
The latter equation recovers the Hamiltonian case when $\theta=0$. If a given Casimir $C( \omega )$ is not sign-definite, one may take its square, which is \emph{still a Casimir}, and find that
\begin{equation}
\frac12\frac{d\,C( \omega )^2}{dt} = 
- \,\theta\,C( \omega )^2 \!\! 
\int_\mathcal{D} \left[ \frac{\delta C}{\delta \omega } , \frac{\delta h}{\delta \omega } \right]
L
\left[ \frac{\delta C}{\delta \omega } , \frac{\delta h}{\delta \omega } \right]
dxdy
= - \,\theta\,C( \omega )^2
\left\|\left[ \frac{\delta C}{\delta \omega }, \frac{\delta h}{\delta \omega }\right]\right\|^2_L
\,
\label{Csquared-dissip}
\end{equation}
for  $ \omega $ verifying the equations \eqref{2LPB-2D} associated to the Casimir function $\frac12C( \omega)^2 $. These calculations have demonstrated the following.

\begin{proposition}\rm
For $\theta>0$ in equation (\ref{2LPB-2D}), the energy Hamiltonian $h( \omega )$ is conserved and squared Casimirs $C( \omega)^2$ decay exponentially at a rate proportional to the square of the (possibly degenerate) norm induced by the $L$-pairing with linear operator $L$. 
\end{proposition}

\begin{remark}\rm
$L$ is called a \emph{scale-selective} operator, because choosing it to contain higher spatial derivatives emphasizes the higher wavenumbers in the $L$-norm in equation (\ref{Csquared-dissip}). This, in turn, determines the Casimir dissipation rate. Thus, the choice of the operator $L$ establishes the range of sizes of spatial scales which most contribute to the modifications of the equations due to Casimir dissipation. In particular, if one were to take 
\[
L=(1-\alpha^2\Delta)^s
\,, \quad\hbox{so that}\quad
\left\|\left[ \frac{\delta C}{\delta \omega }, \frac{\delta h}{\delta \omega }\right]\right\|^2_L = \left\|\left[ \frac{\delta C}{\delta \omega }, \frac{\delta h}{\delta \omega }\right]\right\|^2_{H^s}
,\]
then scale sizes greater or less than the length $\alpha$ would tend to behave differently. 
The scale-selection operator $L$ will allow the resulting modifications of the fluid motion equations to be interpreted in several examples as parameterizing the nonlinear, dynamical interactions among disparate scales that lead to emergent features of the flow.
\end{remark}

\paragraph{Double Lie-Poisson bracket vorticity dynamics (\ref{2LPB-2D}) in 2D.}
The 2D vorticity dynamics generated by the quadratic Lie-Poisson bracket structure in equation (\ref{2LPB-2D}) is found by choosing $f(\omega)=\omega$ above, for which
\begin{align}
\begin{split}
\frac{\partial \omega}{\partial t}  
&= \left[ \frac{\delta h}{\delta \omega }\,,\,\omega \right] 
+\theta \left[ \frac{\delta h}{\delta \omega }\,,\,L
\left[ \frac{\delta h}{\delta \omega }, \frac{\delta C}{\delta \omega } \right] \right]
\\
&= \left[ \frac{\delta h}{\delta \omega }\,,\,\omega 
+\theta\, L
\left[ \frac{\delta h}{\delta \omega }, \frac{\delta C}{\delta \omega } \right]  \right] 
\,,
\end{split}
\label{2LPB-2D-vort-eqn}
\end{align}
after using the identity in (\ref{2D-integparts-id}) obtained from integration by parts. 

\begin{remark}[Double bracket dissipation structure and ``anticipated vorticity'']\rm$\,$

\begin{itemize}
\item
When the operator $L={\rm Id}$, equation (\ref{2LPB-2D-vort-eqn}) defines a \emph{double} bracket dissipation structure whose formula is reminiscent of the double bracket dissipation introduced in \cite{BKMR1996} and \cite{HoPuTr2008} which would yield here the dissipation term $ \theta \left[ \omega , \left[ \omega , \frac{\delta h}{\delta \omega }\right] \right] $. Both the similarities and differences between these two approaches can be clearly seen from their respective abstract Lie algebraic formulations, see Remark \ref{CD_VS_DB} later.
\item
When one sets $\delta h/\delta \omega= \psi$, with stream function $\psi(x,y)$, and chooses the Casimir to be the \emph{enstrophy}, $C=\frac12\int_\mathcal{D} \omega^2\,dxdy$, which is a measure of the turbulence intensity, then equation (\ref{2LPB-2D-vort-eqn}) recovers the \emph{anticipated vorticity model} (AVM) of \cite{SaBa1981,SaBa1985} for an appropriate choice of the operator $L$. See also \cite{VaHu1988}. Namely, in terms of the stream function representation of the 2D vector velocity $\mathbf{u}={\rm curl}\,(\psi\mathbf{\hat{z}})= - \,\mathbf{\hat{z}}\times \nabla \psi$ with vorticity $\omega\,\mathbf{\hat{z}}={\rm curl}\,{\rm curl}\,(\psi\mathbf{\hat{z}})=-\Delta\psi\,\mathbf{\hat{z}}$, equation (\ref{2LPB-2D-vort-eqn}) above with $L={\rm Id}$ and $\delta h/\delta \omega= \psi$ becomes
\begin{align}
\begin{split}
\frac{\partial \omega}{\partial t}  
&= \left[ \psi\,,\,\omega \right] 
+\theta \left[ \psi\,,\,L\left[ \psi, \omega \right] \right]
= \left[ \psi\,,\,\omega 
+\theta L\left[ \psi, \omega \right]  \right] 
\\
&= -\, \mathbf{u}\cdot \nabla \omega 
+ \theta\, \mathbf{u}\cdot \nabla L(\mathbf{u}\cdot \nabla \omega)
= - \,\mathbf{u}\cdot \nabla (\omega - \theta\,L(\mathbf{u}\cdot \nabla \omega))
\,.
\end{split}
\label{AVM-eqn}
\end{align}
The word ``anticipated'' was introduced in \cite{SaBa1981,SaBa1985} when naming AVM, because if one were to take $\theta\simeq\Delta t$, 
where $\Delta t$ is the time step of the numerical method, 
the term proportional to $\theta$ in the last parenthesis of (\ref{AVM-eqn}) would approximate the vorticity at $t+\Delta t$ for the pure Hamiltonian vorticity evolution at linear order in $\Delta t$. The AVM approach has been very intriguing in the modelling of 2D incompressible turbulence and much has been written about it in that literature. These ideas 
and particularly the choice of the time-scale parameter $\theta$ 
have also been discussed again recently in the context of shallow water modelling, see, e.g., \cite{ChGuRi2011,ChGuRi2012,We2012} and Section \ref{RSW-sec} below. 

Other \emph{scale-selective} AVM models considered in \cite{SaBa1981,SaBa1985} were obtained by making an appropriate choice for the differential operator $L$ in equation (\ref{2LPB-2D}), or (\ref{2LPB-2D-vort-eqn}). In particular, \cite{SaBa1981,SaBa1985} 
take $L=R^{16}\Delta^8$,  corresponding to the Sobolev space $H^8$ and thereby introducing a characteristic length-scale $R$ 
and a type of hyperviscosity for the Casimir dissipation. 
This choice biases the effect of the AV term in (\ref{AVM-eqn}) so that $L$ affects predominantly the excitations with higher wave numbers, that is those with $kR>1$, and $L$ rapidly becomes stronger as $k$ increases. In summary, the effect of the AV term is to preserve the kinetic energy of all of the excitations at \emph{all wave numbers}, while dissipating the enstrophy, or turbulent intensity, of the higher wave number modes with $k>1/R$, increasingly at higher wave numbers.  This process takes the higher wave number modes out of play in the forward cascade of enstrophy in 2D turbulence and thus may have a regularizing effect, although we are not aware of any rigorous analysis of the solution behavior for the AVM models. 
\end{itemize}
\end{remark}

\begin{remark}[Regularisation.]\label{alpha-remark1}\rm
Even if AVM does have a regularizing effect in 2D, however, it would be unlikely for any of its 2D regularization properties to carry over to 3D turbulence, because only $L^2$ control of the kinetic energy of the solutions  in the AVM models is afforded in 3D, which is not effective against the vortex stretching term that dominates 3D flows. Controlling the latter requires at least $H^1$ control, as discussed in \cite{FoHoTi2001,FoHoTi2002} for the Navier-Stokes-alpha model. In fact, as we shall see, the dissipation of enstrophy at small scales does leads to a creation of circulation at \emph{all scales}. This may be an indirect mechanism for coarse-graining and thus may provide some regularization by transferring circulation at small scales into circulation at all scales.

\end{remark}

\begin{remark}[Quasigeostrophy]\rm
The flexibility of the Casimir dissipation approach can be shown by passing to the case of 2D rotating quasigeostrophic fluids (QG), for which the vorticity $\omega$ is replaced by \emph{potential} vorticity $q$ satisfying the relation $\delta h/\delta q= \psi$ and equation \eqref{AVM-eqn} becomes
\[
\frac{\partial q}{\partial t}=  - \,\mathbf{u}\cdot \nabla (q - \theta\,L\big(\mathbf{u}\cdot \nabla q)\big), \quad \mathbf{u} = - \mathbf{\hat{z}}\times \nabla \psi , \quad q=- \Delta \psi + \mathcal{F} \psi+f
\,,
\]
where the constants $\mathcal{F}$ and $f$ denote the square of the inverse Rossby radius and the rotation frequency, respectively.

The coarse-graining effect of the AV term in numerical simulations of 2D quasigeostrophic flows was demonstrated in \cite{GrSa1989}, who found that the AV method apparently promoted the inverse cascade and thereby kept energy from piling up at the smallest resolved scales. This coarse-graining effect enhanced their long-time simulations of the emergence of large-scale Fofonoff gyres in weakly decaying 2D quasigeostrophic flow.   
\end{remark}

\paragraph{Double Lie-Poisson bracket vorticity dynamics in 3D.}
Our goal now is to write the quadratic Lie-Poisson bracket structure (\ref{2LPB-2D}) in any number of dimensions and illustrate its effects in several hopefully illuminating examples. 
We begin by using (\ref{brkt-idnty}) in 2D to extend the quadratic vorticity bracket expression (\ref{2LPB-2D}) from 2D to 3D, as follows,
\begin{align}
\begin{split}
\frac{df( \bom )}{dt}
=
\{f\,,\,h\}( \bom )
&= \int_\mathcal{D}\bom\cdot  {\rm curl}\,\frac{\delta f}{\delta \bom } \times {\rm curl}\,\frac{\delta h}{\delta \bom }
\,d^3x
\\
&\hspace{1mm}
- \theta \!\! \int_\mathcal{D}
{\rm curl}\,\frac{\delta f}{\delta \bom } \times {\rm curl}\,\frac{\delta h}{\delta \bom }
\cdot L
\left(
{\rm curl}\,\frac{\delta C}{\delta \bom } \times {\rm curl}\,\frac{\delta h}{\delta \bom }\right)\,d^3x
\,.
\end{split}
\label{2LPB-3D}
\end{align}
The first (skewsymmetric, Lie-Poisson) summand in this modified vorticity bracket may be written in terms of velocity equivalently as 
\begin{align*}
\{f\,,\,h\}_{+}(\mathbf{u})
&= \int_\mathcal{D}\bom\cdot  {\rm curl}\,\frac{\delta f}{\delta \bom } \times {\rm curl}\,\frac{\delta h}{\delta \bom }
\,d^3x
\\&= 
\int_\mathcal{D} \mathbf{u}\cdot  {\rm curl}\left(\frac{\delta f}{\delta \mathbf{u} } 
\times \frac{\delta h}{\delta \mathbf{u} }\right)
\,d^3x
\\&= 
-\int_\mathcal{D} \mathbf{u}\cdot  
\left[\frac{\delta f}{\delta \mathbf{u} } \,,\, \frac{\delta h}{\delta \mathbf{u} }\right]
\,d^3x\,,
\end{align*}
where now in 3D the square brackets $[\,\cdot\,,\,\cdot\,]$ represent Lie brackets of divergence-free vector fields, and we have used the identity $[ \mathbf{u} , \mathbf{v} ]=- \operatorname{curl}( \mathbf{u} \times \mathbf{v} )$ for $ \nabla \cdot \mathbf{u} = 0 = \nabla \cdot \mathbf{v}$. 

The second (inner product) summand in the vorticity bracket (\ref{2LPB-3D}) depends on the Casimir $C$ and may be written in terms of velocity equivalently as 
\begin{align*}
\{f\,,\,h\}_{2}(\mathbf{u})
&=  
- \theta \!\! \int_\mathcal{D}
{\rm curl}\,\frac{\delta f}{\delta \bom } \times {\rm curl}\,\frac{\delta h}{\delta \bom }
\cdot  L 
\left(
{\rm curl}\,\frac{\delta C}{\delta \bom } \times {\rm curl}\,\frac{\delta h}{\delta \bom }\right)\,d^3x
\\
&=  
- \theta \!\! \int_\mathcal{D}
\left(\frac{\delta f}{\delta \mathbf{u} } \times \frac{\delta h}{\delta \mathbf{u} }
\right)\cdot  L 
\left(
\frac{\delta C}{\delta \mathbf{u} } \times \frac{\delta h}{\delta \mathbf{u} }\right)\,d^3x
\\
&=  
- \theta \!\! \int_\mathcal{D}
{\rm curl}\left(\frac{\delta f}{\delta \mathbf{u} } \times \frac{\delta h}{\delta \mathbf{u} }
\right)\cdot {{\rm curl}^{-1}  L \,
{\rm curl}^{-1} {\rm curl}}\left(
\frac{\delta C}{\delta \mathbf{u} } \times \frac{\delta h}{\delta \mathbf{u} }\right)\,d^3x
\\
&=
- \theta \!\! \int_\mathcal{D}
\left[\frac{\delta f}{\delta \mathbf{u} } \,,\, \frac{\delta h}{\delta \mathbf{u} }\right]
\cdot {{\rm curl}^{-1} L \,{\rm curl}^{-1}}
\left[\frac{\delta C}{\delta \mathbf{u} } \,,\, \frac{\delta h}{\delta \mathbf{u} }\right]\,d^3x
\\
&=:
- \theta\,\gamma \left( 
\left[\frac{\delta f}{\delta \mathbf{u} } \,,\, \frac{\delta h}{\delta \mathbf{u} }\right]
\,,
\left[\frac{\delta C}{\delta \mathbf{u} } \,,\, \frac{\delta h}{\delta \mathbf{u} }\right]\right)_\Lambda
.\end{align*}
Here, the positive operator $ \Lambda :=   {\rm curl}^{-1} L\,{\rm curl}^{-1}$ defines a symmetric bilinear form $ \gamma $ on velocities. That is, on the Lie algebra of divergence-free vector fields we have $\gamma: \mathfrak{X}\times \mathfrak{X} \to \mathbb{R}$
\begin{align}
\gamma ( \mathbf{u} , \mathbf{v} )
:=\int_ \mathcal{D} \mathbf{u} \cdot \Lambda \mathbf{v} \,d ^3 x
:= \left\langle   \mathbf{u} ,\, \mathbf{v}^\flat  \right\rangle
\,.
\label{gamma-def}
\end{align}
The second equality here defines a pairing $\langle\,\cdot\,,\,\cdot\, \rangle: \mathfrak{X}\times \mathfrak{X}^* \to \mathbb{R}$, with the flat-operator $\flat: \mathfrak{X}\to \mathfrak{X}^*$ given by $\mathbf{v}^\flat=(\Lambda\mathbf{v})\cdot d \mathbf{x}$.

Having investigated the structure of the modified vorticity bracket for vorticity $\bom$ in equation (\ref{2LPB-3D}), we may now write the modified equation for the vorticity, as
\begin{align}
\begin{split}
\frac{\partial \bom}{\partial t}  
&=
{\rm curl}\left({\rm curl}\frac{\delta h}{\delta \bom }\times \bom\right)
- \theta\, {\rm curl} 
\left({\rm curl}\frac{\delta h}{\delta \bom }\times
L \left( {\rm curl}\,\frac{\delta C}{\delta \bom } \times {\rm curl}\,\frac{\delta h}{\delta \bom } \right)  \right)
\\
&=
{\rm curl}\left[{\rm curl}\frac{\delta h}{\delta \bom } \times 
\left(\bom - \theta\, L \left({\rm curl}\frac{\delta C}{\delta \bom }\times
 {\rm curl}\,\frac{\delta h }{\delta \bom }\right)  \right) \right]
\,.
\end{split}
\label{2LPB-3Dvort}
\end{align}

\paragraph{Conserved quantities in 3D.}
In 3D the fluid kinetic energy is expressed in terms of its vorticity as 
\begin{equation}
h(\bom) 
= \frac12\int_\mathcal{D} |\mathbf{u}|^2  \,d^3x
= \frac12\int_\mathcal{D} \bom\cdot (-\Delta^{-1}\bom)  \,d^3x
\label{erg-def}
\end{equation}
and the only Casimir in the Lie-Poisson Hamiltonian formulation is the helicity 
\begin{equation}
C(\bom) 
= \frac12\int_\mathcal{D} \bom\cdot \mathbf{u}  \,d^3x
= \frac12\int_\mathcal{D} \bom\cdot {\rm curl}^{-1}\bom  \,d^3x
\,,
\label{helicity-def}
\end{equation}
in which we may take ${\rm curl}^{-1}\bom=\mathbf{u}$ to be divergence-free.
The helicity represents the total linkage of the lines of vorticity with itself. It is preserved by the Euler fluid equations, but is dissipated by the viscous Navier-Stokes equations. 
For more details and discussions of the nature of fluid helicity, see \cite{ArKh1998}.
For the choices of $h(\bom)$ and $C(\bom)$ above, we have
\begin{align}
\frac{\delta h}{\delta \mathbf{u} }={\rm curl}\frac{\delta h}{\delta \bom } = \mathbf{u}
\quad\hbox{and}\quad
\frac{\delta C}{\delta \mathbf{u} }={\rm curl}\frac{\delta C}{\delta \bom } = {\rm curl}\,\mathbf{u} = \bom
\,.
\label{helicity-var}
\end{align}
Consequently, the vorticity dynamics equation (\ref{2LPB-3Dvort}) for these choices becomes
\begin{align}
\begin{split}
\frac{\partial \bom}{\partial t} 
&=
{\rm curl}\left(\mathbf{u} \times 
\left(\bom - \theta\,L \left(\bom \times \mathbf{u} \right)  \right) \right)
\\&=
-\,\Big[ \mathbf{u}\,,\, \bom - \theta\,{\rm curl}\left(\Lambda\big[ \mathbf{u}\,,\, \bom \big]\right) \Big]
.\end{split}
\label{2LPB-3Dvort-Leqn}
\end{align}
When the operator $L$ is the identity $(L={\rm Id})$ in these expressions, we have $\Lambda=-\Delta^{-1}$ (the inverse Laplacian) and when $\Lambda={\rm Id}$, we have $L=-\Delta$ (the Laplacian). 

\paragraph{Ertel's theorem.}
In velocity form, the motion equation reads
\begin{equation}\label{3D_Casimir_velocity} 
\partial _t \mathbf{u} + (\boldsymbol{\omega}- \theta L ( \boldsymbol{\omega} \times \mathbf{u} )) \times \mathbf{u} = - \nabla p\,.
\end{equation} 
Note that ${\rm div}(L (\boldsymbol{\omega} \times \mathbf{u} ))=0$, since 
$L={\rm curl}\,\Lambda\,{\rm curl}$ as a composition of operators. 

We may re-arrange the motion equation (\ref{3D_Casimir_velocity}) into the form of Euler's fluid equation with an additional force in \emph{Lorentz form}, so it does no work and thereby conserves kinetic energy: 
\begin{equation}\label{3D_Casimir_velocity-Lorentz} 
\partial _t \mathbf{u} - \mathbf{u} \times \boldsymbol{\omega}
+ \nabla p
= -\, \mathbf{u} \times \mathbf{B}
\,,\end{equation} 
where the divergence-free vector $\mathbf{B}$ is defined by 
\begin{equation}\label{Bvector-def}
\mathbf{B}
:=
\theta\,L ( \boldsymbol{\omega} \times \mathbf{u} )
=
\theta\,{\rm curl}\,\Lambda\,{\rm curl} (  \boldsymbol{\omega} \times \mathbf{u} )  \,.
\end{equation} 
The equation for vorticity dynamics (\ref{2LPB-3Dvort-Leqn}) in this notation becomes
\begin{align}
\begin{split}
\frac{\partial \bom}{\partial t} 
&=
{\rm curl}\left(\mathbf{u} \times 
\left(\bom - \mathbf{B}  \right) \right)
=
-\,\Big[ \mathbf{u}\,,\, \bom - \mathbf{B} \Big]
\\&=
-\, \mathbf{u}\cdot\nabla (\bom - \mathbf{B})
+  (\bom - \mathbf{B})\cdot\nabla\mathbf{u}
,\end{split}
\label{3Dvort-Leqn}
\end{align}
in which the last term is a \emph{modified vortex stretching term}, which either enhances or diminishes vortex stretching, depending on the alignment and relative magnitudes of the vectors $\bom$ and $\mathbf{B}$. An interesting feature of this modified vortex stretching term is its effect on the Ertel theorem for this system \cite{Ertel1942}. Calculating from equation (\ref{3Dvort-Leqn}) for any \textcolor{magenta}{time-dependent}  scalar function $\lambda $ yields
\begin{equation}\label{3D_Ertel} 
\frac{D}{Dt} \big(\bom \cdot\nabla \lambda \big)
= 
\bom \cdot\nabla \frac{D\lambda }{Dt}
+ 
\big[ \mathbf{u}\,,\, \mathbf{B} \big]\cdot\nabla \lambda 
\quad\hbox{with}\quad
\frac{D}{Dt} := \partial _t + \mathbf{u}\cdot\nabla
\,.
\end{equation} 

\begin{remark}[The loss of Ertel's theorem]\rm
The modified vortex stretching term $[ \mathbf{u}\,,\, \mathbf{B}]$ in equation (\ref{3D_Ertel}) implies that Ertel's theorem \emph{does not hold} for the AV models. Namely, conservation on fluid parcels of a quantity $\lambda $ (e.g., potential temperature, as in atmospheric physics, or buoyancy, as in ocean circulation dynamics) does \emph{not} imply conservation on fluid parcels of the corresponding potential vorticity $q=\bom \cdot\nabla \lambda $, as it does in the absence of $\mathbf{B}$ and the modified vortex stretching term $[ \mathbf{u}\,,\, \mathbf{B}]$ in equation (\ref{3Dvort-Leqn}). 
Conservation of potential vorticity on fluid parcels is a fundamental organizing principle for the standard fluid theories of ocean and atmosphere dynamics
\cite{HMR1985,HMc90,HKR1994}.
Its loss for the AV models is an important difference of AV models from the standard theories for ocean and atmosphere dynamics. As we shall see later in Section \ref{KN-thm}, the loss of potential vorticity conservation along flow lines corresponds to the loss of Kelvin's circulation theorem in AV models.
\end{remark}

\begin{remark}[Boundary conditions in 3D]\rm
\label{bdycond-remark}

If $ \mathcal{D} $ has smooth boundary, we consider the Lie algebra $\mathfrak{g}  = \mathfrak{X} _{div}( \mathcal{D} )= \{ \mathbf{u} \in \mathfrak{X}  ( \mathcal{D} )\mid \operatorname{div} \mathbf{u} =0 , \;\; \mathbf{u} \| \partial \mathcal{D}  \}$ and take $ \mathfrak{g}  ^\ast := \mathfrak{g}  $.
That is,  boundary condition $ \mathbf{u} \cdot \mathbf{n} =0$ can be used for Casimir dissipation if we identify the dual Lie algebra with the Lie algebra, i.e. if we use the velocity representation.

$(1)$ The Lie-Poisson part is, as usual, computed as
\begin{align*} 
\{f,g\}_+( \mathbf{u} )
&=
\int_ \mathcal{D} \mathbf{u} \cdot \operatorname{curl} \left( \frac{\delta f}{\delta \mathbf{u} }\times \frac{\delta h}{\delta \mathbf{u} }\right) 
d^3\mathbf{x}
\\
&=
\int_ \mathcal{D} 
\operatorname{div} \left(\mathbf{u} \times 
\left( \frac{\delta f}{\delta \mathbf{u} }\times 
\frac{\delta h}{\delta \mathbf{u} } \right)   \right) 
d^3\mathbf{x}  
+
\int_ \mathcal{D} \operatorname{curl} \mathbf{u} \cdot \left( \frac{\delta f}{\delta \mathbf{u} }\times \frac{\delta h}{\delta \mathbf{u} }\right)  
d^3\mathbf{x}
\\
&=\int_ { \partial \mathcal{D} } \mathbf{u} \times \left( \frac{\delta f}{\delta \mathbf{u} }\times \frac{\delta h}{\delta \mathbf{u} } \right) \cdot \mathbf{n}\, 
d S
+\int_ \mathcal{D} \frac{\delta f}{\delta \mathbf{u} }\cdot\left(  \frac{\delta h}{\delta \mathbf{u} }\times \operatorname{curl} \mathbf{u}\right)  
d^3\mathbf{x}   \\
&= \int_ { \partial \mathcal{D} } \left( \mathbf{u} \cdot \frac{\delta h}{\delta \mathbf{u} }\right) \frac{\delta f}{\delta \mathbf{u} }\cdot \mathbf{n}\,d S - \left( \mathbf{u} \cdot \frac{\delta f}{\delta \mathbf{u} }\right) \frac{\delta h}{\delta \mathbf{u} }\cdot \mathbf{n}\,d S 
+\int_ \mathcal{D} \frac{\delta f}{\delta \mathbf{u} }\cdot\left(  \frac{\delta h}{\delta \mathbf{u} }\times \operatorname{curl} \mathbf{u}\right)  
d^3\mathbf{x} \,.
\end{align*} 
The boundary term vanishes since $ \frac{\delta f}{\delta \mathbf{u} }\cdot \mathbf{n}=0$. So $\dot f=\{f,h\}_+$ yields $ \partial _t \mathbf{u} + \operatorname{curl} \mathbf{u}\times \frac{\delta h}{\delta \mathbf{u} }=- \nabla p$, upon using the Hodge decomposition to include the pressure gradient.

$(2)$ For the dissipative part, one computes
\begin{align*} 
\{f,h\}_2( \mathbf{u} )
&=- \theta \int_ \mathcal{D} \operatorname{curl} \left( \frac{\delta f}{\delta \mathbf{u} }\times \frac{\delta h}{\delta \mathbf{u} }\right) \cdot \Lambda \operatorname{curl} \left( \frac{\delta C}{\delta \mathbf{u} } \times \frac{\delta h}{\delta \mathbf{u} }\right) d^3\mathbf{x} 
\\
&=: - \theta \int_ \mathcal{D} \operatorname{curl} \left( \frac{\delta f}{\delta \mathbf{u} }\times \frac{\delta h}{\delta \mathbf{u} }\right) \cdot \mathbf{A}\,   
d^3\mathbf{x}
\\
&= - \theta \int_\mathcal{D} \operatorname{div} \left(\left( \frac{\delta f}{\delta \mathbf{u} }\times \frac{\delta h}{\delta \mathbf{u} }\right) \times \mathbf{A} \right) d^3\mathbf{x}  - \theta \int_ \mathcal{D} \left( \frac{\delta f}{\delta \mathbf{u} }\times \frac{\delta h}{\delta \mathbf{u} }\right) \cdot \operatorname{curl} \mathbf{A}\, d^3\mathbf{x}  \\
&= - \theta \int_{ \partial \mathcal{D} } \left( \mathbf{A}\times  \left( \frac{\delta h}{\delta \mathbf{u} }\times \frac{\delta f}{\delta \mathbf{u} }\right) \right)  \cdot \mathbf{n} \,d S - \theta \int_ \mathcal{D} \left( \frac{\delta f}{\delta \mathbf{u} }\times \frac{\delta h}{\delta \mathbf{u} }\right) \cdot \operatorname{curl} \mathbf{A}\, d^3\mathbf{x}  \\
&= - \theta \int_{ \partial \mathcal{D} } \left( \left( \mathbf{A} \cdot \frac{\delta f}{\delta \mathbf{u} }\right) \frac{\delta h}{\delta \mathbf{u} }\cdot \mathbf{n} - \left( \mathbf{A} \cdot \frac{\delta h}{\delta \mathbf{u} }\right) \frac{\delta f}{\delta \mathbf{u} }\cdot \mathbf{n} \right)   \, d S - \theta \int_ \mathcal{D} \left( \frac{\delta h}{\delta \mathbf{u} }\times  \operatorname{curl} \mathbf{A} \right) \cdot   \frac{\delta f}{\delta \mathbf{u} }d^3\mathbf{x} \,.
\end{align*}
As in (1), the boundary term again vanishes  and $\dot f=\{f,h\}_++\{f,g\}_2$, for all $f$, yields  
\[
\partial _t \mathbf{u} + \operatorname{curl} \mathbf{u}\times \frac{\delta h}{\delta \mathbf{u} }=\theta \operatorname{curl} \mathbf{A} \times \frac{\delta h}{\delta \mathbf{u} }  - \nabla p
\]
on using the Hodge decomposition to include the pressure gradient.
In addition, we have defined 
\begin{equation}
\mathbf{A} := \Lambda \operatorname{curl} \left( \frac{\delta f}{\delta \mathbf{u} }\times \frac{\delta h}{\delta \mathbf{u} }\right)
\label{Avector-def}
\end{equation}
which is related to the vector $ \mathbf{B} $ in equation \eqref{Bvector-def} by the relation $ \theta \operatorname{curl} \mathbf{A} = \mathbf{B} $. 

Thus, in the velocity representation and using $ \Lambda $ (and not $L$) as the given differential operator, no additional hypotheses are needed about boundary conditions, and the differential operator $ \Lambda $ never needs to be inverted. This applies both in $2D$ and $3D$. Thus, we may freely ignore boundary terms when integrating by parts for fixed boundaries in the velocity representation. We may also ignore such boundary terms when using the equivalent $\boldsymbol{\omega} $-vorticity representation for the dual Lie algebra, in $3D$, because those terms may always be removed by transforming to the equivalent velocity representation.
\end{remark}

\begin{remark}[Stationary solutions and generalized Beltrami flows and Lamb surfaces]\rm 
Stationary solutions of \eqref{3D_Casimir_velocity} satisfy
\[
(\boldsymbol{\omega}- \theta L ( \boldsymbol{\omega} \times \mathbf{u} )) \times \mathbf{u} = - \nabla p
\,.\]
This stationary flow condition means that either the vectors of velocity $\mathbf{u}$ and modified vorticity $\boldsymbol{\omega}- \theta L ( \boldsymbol{\omega} \times \mathbf{u} )$ are parallel, a subcase of which is $\boldsymbol{\omega} \times \mathbf{u} =0$ (Beltrami flows), or that the vectors $\mathbf{u}$ and $\boldsymbol{\omega}- \theta L ( \boldsymbol{\omega} \times \mathbf{u} )$ are both tangent to the same pressure surface $p(\mathbf{x})=$ const, {\it viz}
\begin{equation}
 \mathbf{u} \cdot \nabla p = 0
\quad\hbox{and}\quad
 (\boldsymbol{\omega}- \theta L ( \boldsymbol{\omega} \times \mathbf{u} )) \cdot \nabla p = 0
\,.
\label{Lamb-inv-cond}
\end{equation}
For $\theta=0$, a surface defined by $p(\mathbf{x})=$ const for steady flows would be called a Lamb surface, see, e.g., \cite{Ho2011-GM1}; so surfaces satisfying (\ref{Lamb-inv-cond}) may be called \emph{generalized Lamb surfaces}. The geometry of generalized Beltrami flows and Lamb surfaces may be an interesting research topic, because it addresses the effect of the modified vortex stretching term in equation \eqref{3Dvort-Leqn}. 
\end{remark}

\begin{remark}\rm
Equation (\ref{2LPB-3Dvort-Leqn}) is the 3D version of the anticipated vorticity equation (\ref{AVM-eqn}) in 2D with $L$ replaced by ${\rm curl}\,\Lambda=L\,{\rm curl}^{-1}$. 
Again, the constant $\theta\simeq\Delta t$ has units of  time. 
A ``scale-aware'' version of this anticipated vorticity method in 3D may be obtained by selecting the value of the time-scale parameter $\theta$ based on local mean properties of the solution, instead of the time step of the numerical method. However, then the issue arises of whether the time-scale parameter $\theta$ should be part of the symmetric operator $L$ and, thus, whether its spatial dependence should play a role. This 
is beyond the scope of the present work. 
For a discussion of a numerical algorithm with scale awareness for 2D multiresolution grids, see \cite{ChGuRi2012}.

\end{remark}

\paragraph{Anticipated Kelvin circulation.}
Inverting the curl in equation (\ref{2LPB-3Dvort-Leqn}) yields
\begin{align} 
\begin{split} 
(\partial _t + \mathcal{L}_u)(\mathbf{u} \cdot d\mathbf{x})
&= -dp 
+ \theta\, \mathcal{L}_u\left( \,{\rm curl} ^{-1} L( \boldsymbol{\omega} \times \mathbf{u} )\cdot d \mathbf{x} \right) \\
&= -dp 
+ \theta\, \mathcal{L}_u\left( \,{\rm curl} ^{-1} L\,{\rm curl} ^{-1}( {\rm curl} ( \boldsymbol{\omega} \times \mathbf{u} ))\cdot d \mathbf{x} \right) \\
&= -dp 
+ \theta\,\mathcal{L}_u\,\left(  \Lambda [\mathbf{u} , \boldsymbol{\omega} ]\cdot d \mathbf{x} \right)
, 
\end{split} 
\label{Kel-circ-calc}
\end{align} 
where, we have substituted $ \Lambda := {\rm curl} ^{-1} L\,{\rm curl} ^{-1}$ and used $[ \mathbf{u} , \bom ]=- \operatorname{curl}( \mathbf{u} \times \bom )$.
The symbol $\mathcal{L}_u$ denotes the Lie derivative with respect to the velocity vector field $u=\mathbf{u}\cdot\nabla$, expressible when applied to 1-forms in 3D as 
\begin{equation} 
\mathcal{L}_u (\mathbf{v}\cdot d\mathbf{x})
=
-\,
(\mathbf{u}\times{\rm curl}\,\mathbf{v})\cdot d\mathbf{x}
+
d(\mathbf{u}\cdot\mathbf{v})
\label{Lie-der3D-def}
\,,
\end{equation} 
for a vector $\mathbf{u}$ and a co-vector $\mathbf{v}$.
For an introduction to the use of Lie derivatives in fluid mechanics, see \cite{Ho2011-GM1}. For a more advanced discussion, see \cite{HMR1998}. 

Kevin's circulation theorem now becomes, cf. equation (\ref{2LPB-3Dvort-Leqn})
\begin{equation}\label{Kel-circ-thm}
\begin{aligned}
\frac{d}{dt}\oint_{c(\mathbf{u})} \mathbf{u} \cdot d \mathbf{x}
&=
\oint_{c(\mathbf{u})}(\partial _t + \mathcal{L}_u)(\mathbf{u} \cdot d\mathbf{x})
= -\theta \oint_{c(\mathbf{u})}  \mathbf{u}\times 
{\rm curl}\,\big(\Lambda  [ \mathbf{u} ,\boldsymbol{\omega}]\big)\cdot d \mathbf{x}\\
&= -\,\theta \oint_{c(\mathbf{u})}  \mathbf{u}\times 
L( \boldsymbol{\omega} \times  \mathbf{u})\cdot d \mathbf{x}
= \oint_{c(\mathbf{u})} ( \mathbf{u}\times \mathbf{B})\cdot d \mathbf{x}
\,,
\end{aligned}
\end{equation} 
where $L={\rm curl}\,\Lambda\,{\rm curl}$.
The last term is the source of circulation due to Casimir dissipation. 
  
\paragraph{Variational-derivative expressions.}
Recall that, in terms of the velocity, we have the variational expressions
\[
\frac{\delta h}{\delta \mathbf{u} }= \mathbf{u} \quad \text{and} \quad \frac{\delta C}{\delta \mathbf{u} }= \boldsymbol{\omega}    \quad (C=\hbox{helicity}),
\]
so the last line of the calculation in (\ref{Kel-circ-calc}) may be expressed in terms of variational derivatives as
\begin{align}
\left(\partial _t + \mathcal{L}_{\delta h/\delta \mathbf{u}} \right)(\mathbf{u}  \cdot d\mathbf{x})=-dp 
+ \theta\,\mathcal{L}_{\delta h/\delta \mathbf{u}}\left[  \frac{\delta h}{\delta \mathbf{u} }  , \frac{\delta C}{\delta \mathbf{u} } \right] ^\flat 
,
\label{3D-velocity-eqn}
\end{align}
where the musical symbol  $ \flat : \mathfrak{g}  \rightarrow \mathfrak{g}  ^\ast $ maps vector fields into  1-forms.
This type of expression in terms of variational derivatives will also be obtained by using the general theory developed in the next section.

\begin{proposition}[Squared Casimir dissipation -- Helicity]\label{sign-indef-prop}
\rm
Let the 3D version of equation (\ref{Csquared-dissip}) for squared Casimir dissipation be applied to the square of the helicity, with  
\[
C( \bom)=\frac12\int_\mathcal{D}\bom\cdot {\rm curl}^{-1} \bom\,d^3x
\quad\hbox{and}\quad
h( \bom)=\frac12\int_\mathcal{D}|{\rm curl}^{-1} \bom|^2d^3x
=\frac12\int_\mathcal{D}|\mathbf{u}|^2\,d^3x
\,.\]
Then the squared helicity Casimir dissipates according to the equation
\begin{equation}\label{Helicitysquared-dissip}
\begin{aligned} 
\frac12\frac{d\,C( \bom )^2}{dt} 
&=-\, \theta\,C( \bom )^2\gamma\Big(
 \left[\boldsymbol{\omega}  , \mathbf{u} \right] 
,\left[ \boldsymbol{\omega} , \mathbf{u} \right] \Big)
\\
&=: - \,\theta\,C( \bom )^2
\left\|\big[ \mathbf{u}, \bom\big]\right\|^2_{\Lambda}
= - \,\theta\,C( \bom )^2
\left\| \mathbf{u}\times \bom\right\|^2_{L}
\,.
\end{aligned} 
\end{equation}

\end{proposition}

\begin{proof}
The proof is a direct calculation, starting from equation (\ref{2LPB-3D}).
\begin{align}
\begin{split}
\frac12\frac{d\,C( \bom )^2}{dt} 
&=- \theta C( \bom )^2\!\! \int_\mathcal{D}
{\rm curl}\,\frac{\delta C}{\delta \bom } \times {\rm curl}\,\frac{\delta h}{\delta \bom }
\cdot L
\left(
{\rm curl}\,\frac{\delta C}{\delta \bom } \times {\rm curl}\,\frac{\delta h}{\delta \bom }\right)\,d^3x\\
&=- \theta C( \bom )^2\!\! \int_\mathcal{D}
{\rm curl}^{-1} \left[ {\rm curl}\,\frac{\delta C}{\delta \bom } , {\rm curl}\,\frac{\delta h}{\delta \bom }\right] 
\cdot L\,
{\rm curl}^{-1} \left[ {\rm curl}\,\frac{\delta C}{\delta \bom } , {\rm curl}\,\frac{\delta h}{\delta \bom }\right] \,d^3x\\
&=- \theta C( \bom )^2\!\! \int_\mathcal{D}
 \left[ {\rm curl}\,\frac{\delta C}{\delta \bom } , {\rm curl}\,\frac{\delta h}{\delta \bom }\right] 
\cdot \Lambda  \left[ {\rm curl}\,\frac{\delta C}{\delta \bom } , {\rm curl}\,\frac{\delta h}{\delta \bom }\right] \,d^3x\\
&=- \theta C( \bom )^2\!\! \int_\mathcal{D}
 \left[\boldsymbol{\omega}  , \mathbf{u} \right] 
\cdot \Lambda  \left[ \boldsymbol{\omega} , \mathbf{u} \right] \,d^3x
\quad\hbox{with}\quad
\Lambda := {\rm curl}^{-1}L\, {\rm curl}^{-1}
\\
&=
- \theta\,C( \bom )^2\,\gamma\Big(
 \left[\boldsymbol{\omega}  , \mathbf{u} \right] 
,\left[ \boldsymbol{\omega} , \mathbf{u} \right] \Big)
\\
&=:- \theta C( \bom )^2 \|
 \left[\boldsymbol{\omega}  , \mathbf{u} \right] \|^2_ \Lambda \,.
\end{split}
\end{align}
A shorter proof may be obtained by replacing the functional derivatives directly in the first line to get the equivalent result, 
\[
\frac12\frac{d\,C( \bom )^2}{dt} 
= - \theta C( \bom )^2 \| \boldsymbol{\omega} \times \mathbf{u} \| ^2 _{ L}
\,,
\]
and then substituting the interesting identity that relates the $L$- and $\Lambda$-norms,
\begin{align}
 \| \boldsymbol{\omega} \times \mathbf{u} \| ^2 _{ L}
 = \|\left[\boldsymbol{\omega}  , \mathbf{u} \right] \|^2_ \Lambda
 \,.
 \label{LLnorms-id}
\end{align}
\end{proof}


\section{Selective decay by Casimir dissipation: General theory}
\label{GenTheory-sec}

It is clear that the modification process in the previous section produces the anticipated equations in the cases discussed so far, but despite appearances in those examples, this Casimir dissipation process involves considerably more than merely iterating  the evolution operator. 

\subsection{The modified Lie-Poisson (LP) framework}
To describe the modified Lie-Poisson structure for the general theory, we fix a Lie algebra $ \mathfrak{g}  $ with Lie brackets denoted by $ [\,\cdot\,,\,\cdot\,]$, and consider a space $ \mathfrak{g}  ^\ast $ in weak nondegenerate duality with $ \mathfrak{g}  $, that is, there exists a pairing $ \left\langle\,\cdot\,, \,\cdot\,\right\rangle : \mathfrak{g}  ^\ast \times \mathfrak{g}  \rightarrow \mathbb{R}  $, such that for any $ \xi \in \mathfrak{g}  $, the condition $ \left\langle \mu , \xi \right\rangle = 0$, for all $\mu \in \mathfrak{g}  ^\ast $ implies $ \xi = 0 $ and, similarly, for any $ \mu \in \mathfrak{g}  ^\ast $, the condition $ \left\langle \mu , \xi \right\rangle = 0$ for all $\xi  \in \mathfrak{g}  $ implies $ \mu  = 0 $. Recall that $ \mathfrak{g}  ^\ast $ carries a natural Poisson structure, called the Lie-Poisson structure, given by
\begin{equation}\label{right_LP} 
\{f,h\}_+( \mu )= \left\langle \mu , \left[ \frac{\delta f}{\delta \mu }, \frac{\delta h}{\delta \mu }\right] \right\rangle,
\end{equation} 
(\cite{MaRa1994}) where $f, g \in \mathcal{F} ( \mathfrak{g}  ^\ast )$ are real valued functions defined on $ \mathfrak{g}  ^\ast $ and $ {\delta f}/{\delta \mu }\in \mathfrak{g}  $ denotes the functional derivative, defined through the duality pairing $ \left\langle \,\cdot\,, \,\cdot\,\right\rangle $, by
\[
\left\langle \frac{\delta f}{\delta \mu }, \delta \mu \right\rangle = \left.\frac{d}{d\varepsilon}\right|_{\varepsilon=0} f( \mu + \varepsilon \delta \mu ).
\]
This Poisson bracket is obtained by reduction of the canonical Poisson structure on the phase space $T^*G$ of the Lie group $G$ with Lie algebra $ \mathfrak{g}  $. The symmetry underlying this reduction is given by right translation by $G$ on $T^*G$. In the case of ideal fluid motion, this symmetry corresponds to relabelling symmetry of the Lagrangian in Hamilton's principle. 

Recall that a function $C: \mathfrak{g}  ^\ast \rightarrow \mathbb{R} $ is a \textit{Casimir function} for the Lie-Poisson structure (\ref{right_LP}) if it verifies $\{C, f\}_+=0$ for all functions $f \in \mathcal{F} ( \mathfrak{g}  ^\ast )$ or, equivalently $ \operatorname{ad}^*_ { \frac{\delta C}{\delta \mu }}\mu=0$, for all $ \mu \in \mathfrak{g}  ^\ast $, where $ \operatorname{ad}^*_ \xi: \mathfrak{g}  ^\ast \rightarrow \mathfrak{g}  ^\ast $ is the coadjoint operator defined by $ \left\langle \operatorname{ad}^*_ \xi \mu , \eta \right\rangle = \left\langle \mu , [ \xi , \eta ] \right\rangle $.  

Below, we will denote by $ \gamma _ \mu $ the (possibly $ \mu $-dependent) symmetric bilinear form $ \gamma _\mu : \mathfrak{g}  \times  \mathfrak{g}  \rightarrow \mathbb{R}$. This form is said to be \textit{positive} if
\[
\gamma _\mu ( \xi , \xi ) \geq 0, \quad \text{for all $ \xi \in \mathfrak{g}  $}. 
\]

\begin{definition} Given a Casimir function $C(\mu) $, a positive symmetric bilinear form $ \gamma_\mu  $, and a real number $ \theta >0$, we consider the following modification of the LP (Lie-Poisson) equation to produce the Casimir dissipative LP equation:
\begin{equation}\label{Casimir_dissipation} 
\partial _t \mu + \operatorname{ad}^*_ { \frac{\delta h}{\delta \mu }  } \mu 
= \theta \operatorname{ad}^*_ { \frac{\delta h}{\delta \mu }  } 
 \left[ \frac{\delta C}{\delta \mu }, \frac{\delta h}{\delta \mu }\right] ^\flat  
 =
  -\,\theta \operatorname{ad}^*_ { \frac{\delta h}{\delta \mu }  }
  \left(
  \operatorname{ad}_ { \frac{\delta h}{\delta \mu }  }
  \frac{\delta C}{\delta \mu }  \right) ^\flat
=
  -\,\theta \left(\operatorname{ad}^\dagger_ { \frac{\delta h}{\delta \mu }  }
 \left( \operatorname{ad}_ { \frac{\delta h}{\delta \mu }  }
  \frac{\delta C}{\delta \mu }  \right) \right)^\flat
\,,
\end{equation} 
with notation
\[
 \left(\operatorname{ad}^\dagger_\xi\eta \right)^\flat
:= \operatorname{ad}^*_\xi(\eta^\flat)
\,,
\]
where $ \flat : \mathfrak{g}  \rightarrow \mathfrak{g}  ^\ast $ is the flat operator associated to $ \gamma_\mu  $. That is, the linear form $ \xi ^\flat \in \mathfrak{g}  ^\ast $ is given by $ \xi ^\flat ( \eta  )= \gamma _\mu ( \xi , \eta )$, for all $ \xi , \eta \in \mathfrak{g}  $.
\end{definition}

Note that the flat operator need not be either injective or surjective. Note also that in the equations \eqref{Casimir_dissipation} above, the flat operator is evaluated at $ \mu $. It is important to observe that the modification term depends on both the given Hamiltonian function $h$ and the chosen Casimir $C$.
It is convenient to write \eqref{Casimir_dissipation} as
\begin{equation}\label{Casimir_dissipation_tilde} 
\partial _t \mu + \operatorname{ad}^*_ { \frac{\delta h}{\delta \mu }  } \tilde{\mu} =0,
\end{equation} 
for the modified momentum $\tilde{ \mu }:= \mu+ \theta \left[ \frac{\delta h}{\delta \mu }, \frac{\delta C}{\delta \mu }\right] ^\flat$.

\begin{remark}[Left-invariant case]\rm Recall that the Lie-Poisson structure \eqref{right_LP} is associated to \textit{right} $G$-invariance on $T^*G$. We have made this choice because ideal fluids are naturally right-invariant systems in the Eulerian representation. Other systems, such as rigid bodies, are \textit{left} $G$-invariant. In this case, one obtains the Lie-Poisson brackets $\{f,g\}_-( \mu )= -\left\langle \mu , \left[ \frac{\delta f}{\delta \mu }, \frac{\delta h}{\delta \mu }\right] \right\rangle $ and this leads to the following change of sign in the Casimir dissipative LP equation \eqref{Casimir_dissipation}: 
\begin{equation}\label{anticipated-ad_motion} 
\partial _t \mu - \operatorname{ad}^*_ { \frac{\delta h}{\delta \mu }  } \mu 
= \theta \operatorname{ad}^*_ { \frac{\delta h}{\delta \mu }  } 
 \left[ \frac{\delta C}{\delta \mu }, \frac{\delta h}{\delta \mu }\right] ^\flat  
.
\end{equation} 
\paragraph{The rigid body example.}
For the rigid body case itself, one may identify $\mu$ in (\ref{anticipated-ad_motion} ) with $\bPi\in\mathbb{R}^3$, the body angular momentum, and $\partial H/\partial \bPi = \bOm\in\mathbb{R}^3$ with the body angular velocity, while the Casimir is $C=\frac12|\bPi|^2$. Then, choosing for $ \gamma $ the usual inner product on $ \mathbb{R}  ^3 $, (\ref{anticipated-ad_motion}) becomes
\begin{equation}\label{anticipated-RB_motion} 
\frac{d\bPi}{dt} + \bOm\times\bPi = -\,\theta \,\bOm\times(\bPi\times\bOm)
\,.
\end{equation} 
Upon denoting $\nabla=\partial/\partial \bPi$ so that $\nabla C = \bPi$ and $\nabla H = \bOm$, equation (\ref{anticipated-RB_motion}) for the rigid body implies the following form of the dynamics, which extends that of \cite{Nambu1973} for Hamiltonian dynamics that takes place in $\mathbb{R}^3$ to the case of Casimir dissipation, 
\begin{equation}\label{RB_motion-Nambu} 
\frac{dF}{dt}  = -\,\nabla C\cdot (\nabla F\times \nabla H) 
-\,\theta \,(\nabla C\times \nabla H)\cdot (\nabla F\times \nabla H) 
\,.
\end{equation} 
Consequently, the Casimir is dissipated by this modified rigid body equation when $\theta>0$, since
\begin{equation}\label{RB_motion-disCas} 
\frac{dC}{dt}  = -\,\theta \,|\nabla C\times \nabla H|^2
\,.
\end{equation} 
\end{remark} 

\medskip

\paragraph{Lie-Poisson formulation.} In Lie-Poisson form, the Casimir dissipation equation (\ref{Casimir_dissipation}) reads
\begin{align}
\begin{split}
\label{LP_form} 
\frac{df(\mu) }{dt} 
&= \left\langle \frac{\delta f }{\delta \mu }, \partial _t \mu \right\rangle 
= -\left\langle   \frac{\delta f}{\delta \mu },\operatorname{ad}^*_ { \frac{\delta h}{\delta \mu }  } \mu \right\rangle 
+ \theta \left\langle \frac{\delta f}{\delta \mu },\operatorname{ad}^*_ { \frac{\delta h}{\delta \mu }  } \left[ \frac{\delta C}{\delta \mu }, \frac{\delta h}{\delta \mu }\right] ^\flat \right\rangle 
\\
&= \left\langle \mu , \left[ \frac{\delta f}{\delta \mu } , \frac{\delta h}{\delta \mu } \right] \right\rangle 
- \theta \left\langle \left[ \frac{\delta f}{\delta \mu } , \frac{\delta h}{\delta \mu } \right], \left[ \frac{\delta C}{\delta \mu }, \frac{\delta h}{\delta \mu }\right] ^\flat\right\rangle  
\\
&= \left\langle \mu , \left[ \frac{\delta f}{\delta \mu } , \frac{\delta h}{\delta \mu } \right] \right\rangle 
- \theta \left\langle 
\operatorname{ad}_\frac{\delta h}{\delta \mu } 
\frac{\delta f}{\delta \mu } , 
\left(
\operatorname{ad}_\frac{\delta h}{\delta \mu } 
\frac{\delta C}{\delta \mu } \right)^\flat\right\rangle  
\\
&= \left\{ f,h \right\} _+
- \theta\, \gamma \left( \left[ \frac{\delta f}{\delta \mu } , \frac{\delta h}{\delta \mu } \right], \left[ \frac{\delta C}{\delta \mu } , \frac{\delta h}{\delta \mu } \right]\right),
\end{split}
\end{align} 
for an arbitrary function $ f: \mathfrak{g}  ^\ast \rightarrow \mathbb{R}  $.

The energy is preserved, since we have
\[
\frac{dh(\mu) }{dt}= \left\{ h,h \right\} _+\,- \theta\, \gamma \left( \left[ \frac{\delta h}{\delta \mu } , \frac{\delta h}{\delta \mu } \right], \left[ \frac{\delta C}{\delta \mu } , \frac{\delta h}{\delta \mu } \right]\right) =0.
\]
However, when $ \theta >0$ the Casimir function $C$ is dissipated since
\begin{equation}\label{Casimirdissipation} 
\frac{dC(\mu) }{dt}
= \left\{ C,h \right\} _+ -\, \theta\, \gamma \left( \left[ \frac{\delta C}{\delta \mu } , \frac{\delta h}{\delta \mu } \right], \left[ \frac{\delta C}{\delta \mu } , \frac{\delta h}{\delta \mu } \right]\right) 
=-\, \theta \left \| \left[ \frac{\delta C}{\delta \mu } , \frac{\delta h}{\delta \mu } \right]\right \| _\gamma ^2 ,
\end{equation} 
where $\| \xi \|^2_\gamma  :=\gamma_\mu  ( \xi , \xi )$ is the quadratic form (possibly degenerate) associated to the positive bilinear form $ \gamma_\mu  $.

In the left invariant case, the Lie-Poisson part of the selective decay Casimir dissipation form changes by $\{ f,h\}_+\to\{ f,h\}_-$, and in comparison with \eqref{LP_form} one finds,
\begin{equation}\label{leftinv-Cdissip} 
\frac{df(\mu) }{dt}=\left\{ f,h \right\} _- - \theta \gamma \left( \left[ \frac{\delta f}{\delta \mu } , \frac{\delta h}{\delta \mu } \right], \left[ \frac{\delta C}{\delta \mu } , \frac{\delta h}{\delta \mu } \right]\right).
\end{equation}

\begin{remark}[Casimir dissipation versus anticipated setting]\label{Casimir_Diss_VS_AV}\rm 
Clearly, equation (\ref{Casimir_dissipation}) is not equivalent to a simple iteration of the evolution operator, as might be considered in an ``anticipated'' setting in the form
\begin{equation}\label{anticipated_motion} 
\partial _t \mu + \operatorname{ad}^*_ { \frac{\delta h}{\delta \mu }  } \mu 
= \theta \operatorname{ad}^*_ { \frac{\delta h}{\delta \mu }  } 
 \left(\operatorname{ad}^*_ { \frac{\delta h}{\delta \mu }  } \mu\right) 
\,.\end{equation} 
However, the ``anticipated''  modification does not guarantee selective decay of Casimirs, $dC/dt<0$,  in general. The dynamics of equations (\ref{Casimir_dissipation}) for Casimir dissipation and (\ref{anticipated_motion}) for ``anticipated'' motion are only equivalent when 
\begin{equation}\label{equiv-case} 
 \operatorname{ad}^*_ { \frac{\delta h}{\delta \mu }  } 
 \left[ \frac{\delta C}{\delta \mu }, \frac{\delta h}{\delta \mu }\right] ^\flat  
 = \operatorname{ad}^*_ { \frac{\delta h}{\delta \mu }  } 
 \left(\operatorname{ad}^*_ { \frac{\delta h}{\delta \mu }  } \mu\right) 
,
\end{equation} 
which can only hold for quadratic Casimirs $C$. Thus, ``anticipated'' motion and selective Casimir decay can only coincide for certain quadratic Casimirs.
Such a coincidence arises for example in the case of 2D incompressible ideal fluids. For 3D fluids, however, they are distinct, and the anticipated setting \eqref{anticipated_motion} does not yield selective decay of Casimirs.
\end{remark}

\begin{remark}[Relation of the present work to \cite{VaCaYo1989}]\label{Ch-energy dissipation}\rm
Suppose  we had chosen to dissipate energy instead of the Casimir of the Lie-Poisson bracket for fluids. This could be accomplished easily, by exchanging $h\leftrightarrow C$ in the second term of \eqref{LP_form}, thereby obtaining
\begin{align}\label{LP_form2} 
\frac{df(\mu) }{dt} 
&= \left\{ f,h \right\} _+
- \theta\, \gamma \left( 
\left[ \frac{\delta f}{\delta \mu } , \frac{\delta C}{\delta \mu } \right], \left[ \frac{\delta h}{\delta \mu } , \frac{\delta C}{\delta \mu } \right]\right).
\end{align} 
Then, instead of the equation of motion \eqref{Casimir_dissipation} leading to Casimir dissipation,
\begin{equation}\label{Casimir_dissipation2} 
\partial _t \mu + \operatorname{ad}^*_ { \frac{\delta h}{\delta \mu }  } \mu 
= \theta \operatorname{ad}^*_ { \frac{\delta h}{\delta \mu }  } 
 \left[ \frac{\delta C}{\delta \mu }, \frac{\delta h}{\delta \mu }\right] ^\flat  
\,,
\end{equation} 
one would obtain the equation of motion 
\begin{equation}\label{unanticipated-ad_motion} 
\partial _t \mu + \operatorname{ad}^*_ { \frac{\delta h}{\delta \mu }  } \mu 
= -\,\theta \operatorname{ad}^*_ { \frac{\delta C}{\delta \mu }  } 
 \left[ \frac{\delta C}{\delta \mu }, \frac{\delta h}{\delta \mu }\right] ^\flat  
\,.
\end{equation} 

Taking $f=h $ in \eqref{LP_form2} shows that when $ \theta >0$ the energy $h$ is dissipated, since for that choice we have
\begin{equation}\label{erg-dissipation} 
\frac{dh(\mu) }{dt}
= \left\{ h,h \right\} _+ -\, \theta\, \gamma \left( \left[ \frac{\delta C}{\delta \mu } , \frac{\delta h}{\delta \mu } \right], \left[ \frac{\delta C}{\delta \mu } , \frac{\delta h}{\delta \mu } \right]\right) 
=-\, \theta \left \| \left[ \frac{\delta C}{\delta \mu } , \frac{\delta h}{\delta \mu } \right]\right \| _\gamma ^2 ,
\end{equation} 
and the symmetry between $h$ and $C$ is evident again, because the rate of dissipation is the same as before, even though it is now energy rather than Casimir that is being dissipated. This remark relates the basic ideas in \cite{VaCaYo1989} to the present work. 
\end{remark}

\begin{remark}[Left invariant case]
For the left-invariant case, equation \eqref{unanticipated-ad_motion}  becomes
\begin{equation}\label{unanticipated-ad_motion-left} 
\partial _t \mu - \operatorname{ad}^*_ { \frac{\delta h}{\delta \mu }  } \mu 
= -\,\theta \operatorname{ad}^*_ { \frac{\delta C}{\delta \mu }  } 
 \left[ \frac{\delta C}{\delta \mu }, \frac{\delta h}{\delta \mu }\right] ^\flat  
\,.
\end{equation} 
In the rigid body example, this is
\begin{equation}\label{unanticipated-RB_motion} 
\frac{d\bPi}{dt} + \bOm\times\bPi 
= 
\theta \,\bPi\times(\bPi\times\bOm)
\,.
\end{equation} 
Hence, the energy is dissipated by this modified rigid body equation when $\theta>0$, since
\begin{equation}\label{RB_motion-disCas} 
\frac{d}{dt}\frac12\bOm\cdot\bPi  
= \theta \,\bOm\cdot\bPi\times(\bPi\times\bOm)
= -\theta \,|\bOm\times\bPi|^2
\,.
\end{equation} 
\end{remark}

\begin{remark}[Casimir dissipation \emph{versus} double bracket dissipation]\label{CD_VS_DB}\rm The Casimir dissipative LP setting considered here is essentially different from the double bracket dissipation setting. Indeed, double bracket equations dissipate energy while they preserve the Casimirs. This is exactly the opposite of what happens in the present setting. There are however apparent similarities in the Lie algebraic formulations. Indeed, for general Lie algebras the double bracket dissipation equations can be written as
\begin{equation}\label{Lie_algebraic_DB} 
\frac{d f( \mu )}{dt}= \{f,h\}_+( \mu )-\theta  \gamma ^\ast \left( \operatorname{ad}^*_{ \frac{\delta k}{\delta \mu }} \mu ,  \operatorname{ad}^*_{ \frac{\delta f}{\delta \mu }} \mu  \right) ,
\end{equation} 
(compare with equation \eqref{LP_form}) where $ \gamma $ is a inner product on $ \mathfrak{g}  $, $ \gamma ^\ast $ is the inner product induced on $ \mathfrak{g}  ^\ast $, and $k: \mathfrak{g}  ^\ast \rightarrow \mathbb{R}  $ is a given function. One readily checks that Casimirs are preserved while, in the special case $k=h$, the energy dissipates. In that case, the equation of motion arising from the double bracket in (\ref{Lie_algebraic_DB}) is given by
\begin{equation}\label{DB-motioneqn} 
\partial _t \mu + \operatorname{ad}^*_ { \frac{\delta h}{\delta \mu }  } \mu 
= \theta \operatorname{ad}^*_{\big(  \operatorname{ad}^\ast_{ \frac{\delta h}{\delta \mu }} \mu\big) ^\sharp}  \mu  
\,,
\end{equation} 
see \cite{BKMR1996,HoPuTr2008}, where $\sharp: \mathfrak{g}  ^\ast \rightarrow \mathfrak{g}  $ is the sharp operator associated to $ \gamma $. 
Formula \eqref{DB-motioneqn} for double bracket dissipation may be compared with the corresponding equation for Casimir dissipation in \eqref{Casimir_dissipation}
\begin{equation}\label{Casimir_dissipation1} 
\partial _t \mu + \operatorname{ad}^*_ { \frac{\delta h}{\delta \mu }  } \mu 
= \theta \operatorname{ad}^*_ { \frac{\delta h}{\delta \mu }  } 
 \left[ \frac{\delta C}{\delta \mu }, \frac{\delta h}{\delta \mu }\right] ^\flat  
,
\end{equation} 
and also with the corresponding equation for ``anticipated motion'' in (\ref{anticipated_motion})
\begin{equation}\label{anticipated_motion1} 
\partial _t \mu + \operatorname{ad}^*_ { \frac{\delta h}{\delta \mu }  } \mu 
= \theta \operatorname{ad}^*_ { \frac{\delta h}{\delta \mu }  } 
 \left(\operatorname{ad}^*_ { \frac{\delta h}{\delta \mu }  } \mu\right) 
.\end{equation} 
The last two of these formulas happen to coincide with \eqref{AVM-eqn} \textcolor{magenta}{(with $L=id$)} for the case of 2D incompressible flow, see Remark \ref{Casimir_Diss_VS_AV}. More generally, this happens whenever an $ \operatorname{Ad}$-invariant inner product $\bar \gamma $ can be chosen on  the Lie algebra $ \mathfrak{g}  $ so that one can identify $ \mathfrak{g}  ^\ast $ with $ \mathfrak{g}$ using this inner product and obtain $ \operatorname{ad}^*_ \mu \xi =[ \xi , \mu ]$. Upon choosing the quadratic Casimir $ C( \mu )= \frac{1}{2} \bar \gamma ( \mu , \mu )$ and the inner product $ \gamma = \bar \gamma $, one observes that \eqref{Casimir_dissipation1} recovers \eqref{anticipated_motion1}.  

The double bracket equations \eqref{DB-motioneqn}, on the other hand, coincide in some particular cases with equations \eqref{unanticipated-ad_motion}, obtained from our approach by exchanging the functions $C$ and $h$. More precisely, this happens in the same situation as before, when an $ \operatorname{Ad}$-invariant  inner product $\bar \gamma $ can be chosen on the Lie algebra.
\end{remark}

\subsection{Kelvin-Noether theorem.}\label{KN-thm}
The well-known Kelvin circulation theorems for standard ideal fluid models  can be seen as reformulations of Noether's theorem and, therefore, they have an abstract Lie algebraic formulation (the Kelvin-Noether theorems), see \cite{HMR1998}. We now examine how the dissipation term introduced above modifies the abstract Kelvin circulation theorem.   

In order to formulate the Kelvin-Noether theorem, one has to choose a manifold $ \mathcal{C} $ on which the group $G$ acts on the left and consider a $G$-equivariant map $ \mathcal{K} : \mathcal{C} \rightarrow \mathfrak{g}  ^{**}$, i.e. $\left\langle  \mathcal{K} (gc),\operatorname{Ad}_{g^{-1} } ^\ast   \nu  \right\rangle =\left\langle \mathcal{K} (c ), \nu  \right\rangle, \forall\; g \in G$. Here $gc$ denotes the action of $g \in G$ on $ c \in \mathcal{C}$ and $ \operatorname{Ad}^*_g$ denotes the coadjoint action defined by $\left\langle \operatorname{Ad}^*_g \mu , \xi \right\rangle = \left\langle \mu \operatorname{Ad}_g \xi \right\rangle $, where $ \mu \in \mathfrak{g}  ^\ast $, $ \xi \in \mathfrak{g}  $, and $\operatorname{Ad}_g$ is the adjoint action of $G$ on $ \mathfrak{g}  $. Given $c \in \mathcal{C}$ and $ \mu \in \mathfrak{g}  ^\ast $, we will refer to $ \left\langle \mathcal{K} (c), \mu \right\rangle $ as the \textit{Kelvin-Noether quantity} (\cite{HMR1998}). In application to fluids, $ \mathcal{C} $ is the space of loops in the fluid domain and $ \mathcal{K} $ is the circulation around this loop, namely
\begin{equation}
\left\langle \mathcal{K} (c), \mathbf{u} \cdot d \mathbf{x} \right\rangle :=\oint_ c \mathbf{u} \cdot d\mathbf{x}.
\end{equation}

The Kelvin-Noether theorem for Casimir dissipative LP equations is formulated as follows. 

\begin{proposition}\label{KN_casimir}  Fix $ c _0 \in \mathcal{C} $ and consider a solution $ \mu (t)$ of the Casimir dissipative LP equation \eqref{Casimir_dissipation}. Let $g(t) \in G$ be the curve determined by the equation $ \frac{\delta h}{\delta \mu }= \dot g g ^{-1} $, $g(0)=e$. Then the time derivative of the Kelvin-Noether quantity $\left\langle \mathcal{K} ( g (t) c _0 ), \mu (t) \right\rangle$ associated to this solution is
\begin{equation}
\frac{d}{dt} \left\langle \mathcal{K} ( g (t) c _0 ), \mu (t) \right\rangle =\theta \left\langle \mathcal{K} (g(t) c _0 ), \operatorname{ad}^*_ { \frac{\delta h}{\delta \mu }  } 
 \left[ \frac{\delta C}{\delta \mu }, \frac{\delta h}{\delta \mu }\right] ^\flat\right\rangle.
\end{equation}
\end{proposition} 
\begin{proof} The proof is a direct calculation. We have
\begin{align*} 
\frac{d}{dt} \left\langle \mathcal{K} ( g (t) c _0 ), \mu (t) \right\rangle &= \frac{d}{dt} \left\langle \mathcal{K} ( c _0 ), \operatorname{Ad}^*_{g(t)} \mu (t) \right\rangle  = \left\langle \mathcal{K} (c _0 ), \operatorname{Ad}^*_{g(t)} ( \partial _t \mu + \operatorname{ad}^*_{ \frac{\delta h}{\delta \mu }} \mu )   \right\rangle \\
&=\theta \left\langle \mathcal{K} (g(t) c _0 ), \operatorname{ad}^*_ { \frac{\delta h}{\delta \mu }  } 
 \left[ \frac{\delta C}{\delta \mu }, \frac{\delta h}{\delta \mu }\right] ^\flat\right\rangle, 
\end{align*} 
where, at the first equality we used the $G$-equivariance of $ \mathcal{K} $ and at the second equality we used the formula $ \frac{d}{dt} \operatorname{Ad}^*_{g(t)}\mu (t) =\operatorname{Ad}^*_{g(t)} ( \partial _t \mu(t)  + \operatorname{ad}^*_{ \dot g(t)  g(t)  ^{-1} } \mu (t) )$, see, e.g., \cite{MaRa1994}.
\end{proof}

Note that $g(t)\in G$ is the motion in Lagrangian coordinates associated to the evolution of the momentum $ \mu (t)\in \mathfrak{g}  ^\ast $ in Eulerian coordinates. The $\theta$ term is an extra source of circulation with a double commutator. This term is absent in the ordinary Lie-Poisson case (i.e., for $ \theta =0$) and therefore in this case the Kelvin-Noether quantity $\left\langle \mathcal{K} ( g (t) c _0 ), \mu (t) \right\rangle$ is conserved along solutions.

\begin{remark}[Three dimensional ideal flows]\rm The Casimir dissipation approach for 3D flows developed in \S\ref{Casimir_dissipation} follows from the present abstract formulation by choosing the Lie algebra $ \mathfrak{g}  = \mathfrak{X}  _{div}( \mathcal{D} )$ of divergence free vector fields on $ \mathcal{D} $, the helicity Casimir $C$, and the positive symmetric bilinear form $ \gamma ( \mathbf{u} , \mathbf{v} ) =\int_{ \mathcal{D} } \mathbf{u} \cdot \Lambda \mathbf{v} d ^3 x$, where $ \Lambda = {\rm curl} ^{-1} \,L\ {\rm curl} ^{-1} $ with $L$ an arbitrary self-adjoint positive operator. Using $\mathbf{v}^\flat = \Lambda \mathbf{v}\cdot\,d\mathbf{x}$, one readily checks that for these choices, equations \eqref{Casimir_dissipation} and \eqref{LP_form} yield \eqref{2LPB-3D} and \eqref{2LPB-3Dvort-Leqn}. Note that Proposition \ref{KN_casimir} applied to this case yields the circulation theorem
\[
\frac{d}{dt} \oint_{c( \mathbf{u} )} \mathbf{u} \cdot d \mathbf{x} = \theta \int_{ c( \mathbf{u} )}\mathcal{L} _ \mathbf{u} \left( \Lambda \left[ \mathbf{u} , \boldsymbol{\omega} \right] \cdot d \mathbf{x} \right) 
= -\theta  \oint_{ c( \mathbf{u} )} \mathbf{u} \times L( \boldsymbol{\omega} \times \mathbf{u})  \cdot d \mathbf{x} 
=
 \oint_{ c( \mathbf{u} )} (\mathbf{u} \times \mathbf{B})  \cdot d \mathbf{x}
\,,
\]
which recovers our previous expression for Kelvin's circulation theorem with Casimir dissipation obtained in \eqref{Kel-circ-thm} by a direct calculation. Rewriting this equation as 
\begin{equation}
\frac{d}{dt} \oint_{c( \mathbf{u} )}  \hspace{-4mm}
\mathbf{u} \cdot d \mathbf{x} 
=
 \int\!\!\!\!\int_{\partial S = c( \mathbf{u} )} \hspace{-8mm}
 \operatorname{curl}(\mathbf{u} \times \mathbf{B})  \cdot d \mathbf{S}
=
 \int\!\!\!\!\int_{\partial S = c( \mathbf{u} )} \hspace{-8mm} 
\big[ \mathbf{u}\,,\, \mathbf{B} \big] \cdot d \mathbf{S}
\,,
\label{vortstretchcirc}
\end{equation}
shows that the modified vortex stretching term in equation \eqref{3Dvort-Leqn} may either enhance or diminish circulation. 
\end{remark}

\subsection{Lagrange-d'Alembert variational principle.} We now explain how the Casimir dissipative LP equations can be obtained via a variational principle. Consider the Lagrangian $\ell: \mathfrak{g}  \rightarrow \mathbb{R}  $ related to $h$ via the Legendre transform, that is, we have
\[
h( \mu )= \left\langle \mu , \xi \right\rangle - \ell( \xi ), \quad  \mu := \frac{\delta \ell}{\delta \xi},
\]
where we assumed that the second relation yields a bijective correspondence between $ \xi $ and $ \mu $.
In terms of  $\ell$, equation \eqref{Casimir_dissipation} for Casimir dissipation reads
\begin{equation}\label{EP_Casimir} 
\partial _t \mu + \operatorname{ad}^*_ { \xi  } \mu 
= \theta \operatorname{ad}^*_ { \xi  } 
 \left[ \frac{\delta C}{\delta \mu }, \xi \right] ^\flat  
\,,\quad \mu := \frac{\delta \ell}{\delta \xi}. 
\end{equation} 
When $ \theta =0$ we recover the \textit{Euler-Poincar\'e equations}
\[
\partial_t \frac{\delta \ell}{\delta \xi}+ \operatorname{ad}^*_\xi  \frac{\delta \ell}{\delta \xi} =0
\]
and it is well-known that these equations can be obtained via the \textit{constrained  variational principle} \cite{ArKh1998,HMR1998}
\[
\delta \int_0 ^T  \ell( \xi ) dt=0, \quad\text{for variations} \quad\delta \xi = \partial _t \zeta -[ \xi , \zeta ],
\]
where $ \zeta \in \mathfrak{g}  $ is an arbitrary curve vanishing at $t=0,T$.

External forces $f( \xi ) \in \mathfrak{g}  ^\ast $ can be included in the Euler-Poincar\'e equations by using the \textit{reduced Lagrange-d'Alembert principle} \cite{Bl2004}
\begin{equation}\label{LdA_EP} 
\delta\left[  \int_0 ^T \ell( \xi ) dt\right] + \int_0 ^T \left\langle f( \xi ), \zeta \right\rangle dt  =0, \quad\text{for variations} \quad\delta \xi = \partial _t \zeta -[ \xi , \zeta ],
\end{equation}
where $ \zeta \in \mathfrak{g}  $ is an arbitrary curve vanishing at $t=0,T$.
Substituting into \eqref{LdA_EP} the formula for the force appearing in the Casimir dissipation equations \eqref{EP_Casimir} recovers equations \eqref{EP_Casimir}, now rederived from the Lagrange-d'Alembert variational principle
\[
\delta\left[  \int_0 ^T \ell( \xi ) dt\right] + \theta \int_0 ^T \gamma \left( \left[ \frac{\delta C}{\delta \mu }, \xi \right] , [ \xi ,\zeta] \right) dt  =0, \quad\text{for variations} \quad\delta \xi = \partial _t \zeta -[ \xi , \zeta ].
\]
Thus, in the Lagrange-d'Alembert formulation, the modification of the motion equation to impose selective decay is seen as an energy-conserving constraint force.

For example, for three dimensional ideal flows, equations \eqref{2LPB-3Dvort-Leqn} can be obtained from the variational principle
\[
\delta \int_0^T\int_ \mathcal{D} \frac{1}{2} | \mathbf{u} | ^2\, d ^3 x\,dt + \theta \int_0^T\int_ \mathcal{D} \left( \left[ \boldsymbol{\omega} , \mathbf{u} \right] \cdot \Lambda \left[ \mathbf{u} , \mathbf{v} \right] \right) \,d ^3 x \,dt=0, \quad \text{for variations} \quad \delta \mathbf{u} = \partial _t \mathbf{v} +[ \mathbf{u} , \mathbf{v} ].
\]
The second term on the left hand side can be also written as
\[
\theta \int_0^T\int_ \mathcal{D}( \boldsymbol{\omega} \times \mathbf{u})  \cdot L( \mathbf{u} \times  \mathbf{v})\,d ^3 x \,dt
\,.
\]
This is the time integral of the rate of virtual work done 
by the force for a virtual motion $\mathbf{v}= \delta \varphi \circ \varphi ^{-1} $ of the fluid configuration $\varphi : \mathcal{D} \rightarrow \mathcal{D} $.


\section{Semidirect product examples}\label{semidirect-sec}

The Hamiltonian structure of fluids that possess advected quantities such as heat, mass, buoyancy, magnetic field, etc., can be understood by using Lie-Poisson brackets for semidirect-product actions of Lie groups on vector spaces.

In this setting, besides the Lie group configuration space $G$, one needs to include a vector space $V $ on which $G$ acts linearly. Its dual vector space $ V ^\ast $ contains the advected quantities. From this, one considers the semidirect product $G \,\circledS\, V$ with Lie algebra $ \mathfrak{g}  \,\circledS\, V$, and the Hamiltonian structure is given by the Lie-Poisson bracket \eqref{right_LP}, written on $ ( \mathfrak{g}  \,\circledS\, V) ^\ast $ instead of $ \mathfrak{g}  ^\ast $. We refer to \cite{HMR1998} for a detailed treatment. Given a Hamiltonian function $h=h( \mu ,a)$ on $( \mathfrak{g}  \,\circledS\, V) ^\ast $ one thus obtains the Lie-Poisson equations
\[
\partial _t (\mu ,a)+ \operatorname{ad}^*_ { \left( \frac{\delta h}{\delta \mu } , \frac{\delta h}{\delta a} \right) } \mu=0,
\]
for $\mu (t)\in \mathfrak{g}  ^\ast $ and $a (t)\in V ^\ast $. More explicitly, making use of the expression of the $ \operatorname{ad}^*$-operator in the semidirect product case, these equations read
\begin{equation}\label{SDP_LP} 
\partial _t \mu + \operatorname{ad}^*_{ \frac{\delta h}{\delta \mu }} \mu   +\frac{\delta h}{\delta a}\diamond a =0, \quad \partial _t a + a \frac{\delta h}{\delta \mu } =0,
\end{equation} 
where the operator $ \diamond : V \times V^\ast \rightarrow \mathfrak{g}  ^\ast $ is defined by 
\begin{equation}\label{diamond-def} 
\left\langle v \diamond a, \xi \right\rangle = -\left\langle a \xi ,v\right\rangle 
,
\end{equation}  
and $ a \xi\in V ^\ast  $ denotes the (right) Lie algebra action of $ \xi \in \mathfrak{g}  $ on $ a \in V ^\ast $.

We can easily extend the Casimir dissipation 
approach of \S\ref{GenTheory-sec} to the semidirect product case. Fixing a Casimir function $C=C( \mu , a)$ for the Lie-Poisson bracket on the semidirect product and a (possibly $( \mu , a)$-dependent) positive symmetric bilinear form $ \gamma _{( \mu ,a)}: \mathfrak{g}  \times \mathfrak{g}  \rightarrow \mathbb{R}$, we 
extend the semidirect-product Lie-Poisson system \eqref{SDP_LP} naturally to allow for Casimir dissipation by setting
\begin{equation}\label{SDP_LP_dissip} 
\partial _t \mu + \operatorname{ad}^*_{ \frac{\delta h}{\delta \mu }} \mu   +\frac{\delta h}{\delta a}\diamond a =\theta \operatorname{ad}^*_ { \frac{\delta h}{\delta \mu }  } 
 \left[ \frac{\delta C}{\delta \mu }, \frac{\delta h}{\delta \mu }\right] ^\flat, \quad \partial _t a + a \frac{\delta h}{\delta \mu } =0,
\end{equation} 
where $ \flat : \mathfrak{g}  \rightarrow\mathfrak{g}  ^\ast $ is the flat operator associated to $ \gamma $.
In Lie-Poisson form, this equation reads
\[
\frac{df(\mu, a) }{dt}=\left\{ f,h \right\} _+(\mu, a)
- \theta\, \gamma \left( \left[ \frac{\delta f}{\delta \mu } , \frac{\delta h}{\delta \mu } \right], \left[ \frac{\delta C}{\delta \mu } , \frac{\delta h}{\delta \mu } \right]\right),
\]
which shows that the energy is conserved while the Casimir dissipates as in \eqref{Casimirdissipation}. 

\paragraph{Kelvin-Noether theorem.}
To formulate the Kelvin-Noether theorem, we consider a $G$-equivariant map $\mathcal{K} : \mathcal{C} \times V ^\ast \rightarrow \mathfrak{g}  ^{\ast\ast} $,  i.e. $\left\langle  \mathcal{K} (gc, a g ^{-1} ),\operatorname{Ad}_{g^{-1} } ^\ast   \nu  \right\rangle =\left\langle \mathcal{K} (c,a ), \nu  \right\rangle, \forall\; g \in G$. Given $ c _0 \in \mathcal{C} $ and a solution $ \mu (t)$, $ a (t) $ of \eqref{SDP_LP_dissip}, the associated Kelvin-Noether quantity reads $\left\langle \mathcal{K} ( g(t) c _0 , a (t) ), \mu (t) \right\rangle $ and computations similar to those in Proposition \ref{KN_casimir} yield the Kevin-Noether theorem
\begin{equation}\label{KN_SDP} 
\frac{d}{dt} \left\langle \mathcal{K} ( g (t) c _0 , a (t) ), \mu (t) \right\rangle = \left\langle \mathcal{K} (g(t) c _0, a (t)  ), \theta\operatorname{ad}^*_ { \frac{\delta h}{\delta \mu }  } 
 \left[ \frac{\delta C}{\delta \mu }, \frac{\delta h}{\delta \mu }\right] ^\flat-\frac{\delta h}{\delta a}\diamond a\right\rangle.
\end{equation} 

\begin{remark}[Natural generalization]\rm Note that the Casimir dissipative LP equation \eqref{SDP_LP_dissip} was obtained from the Lie-Poisson equations on the semidirect product by modifying the momentum  $ \mu \in \mathfrak{g}  ^\ast $ only, while keeping the quantity $a\in V ^\ast $ unchanged. From the Lie algebraic point of view, however, the direct generalization of \eqref{Casimir_dissipation} to semidirect product Lie groups would be
\begin{equation}\label{Casimir_dissipation_SDP} 
\partial _t (\mu ,a)+ \operatorname{ad}^*_ { \left( \frac{\delta h}{\delta \mu } , \frac{\delta h}{\delta a} \right) } \mu = \theta \operatorname{ad}^*_ {\left(  \frac{\delta h}{\delta \mu } , \frac{\delta h}{\delta a} \right) } \left( \left[ \left( \frac{\delta C}{\delta \mu },\frac{\delta C}{\delta a }\right) , \left( \frac{\delta h}{\delta \mu },  \frac{\delta h}{\delta a } \right) \right] ^\flat  \right),
\end{equation} 
where the flat operator $ \flat : \mathfrak{g}  \times V \rightarrow \mathfrak{g}  ^\ast \times V ^\ast $ is associated to a positive symmetric bilinear map $ \gamma_{( \mu , a)} :(\mathfrak{g}  \times V) \times ( \mathfrak{g}  \times V) \rightarrow \mathbb{R}$. This results in a modification of \emph{both} $ \mu $ and $ a$ as
\[
\widetilde \mu = \mu - \theta \left[ \frac{\delta C}{\delta \mu }, \frac{\delta h}{\delta \mu }\right] ^\flat , \quad \widetilde a= a- \theta  \left(\frac{\delta C}{\delta a} \frac{\delta h}{\delta \mu }-\frac{\delta h}{\delta a} \frac{\delta C}{\delta \mu }  \right) ^\flat.
\]
Of course, as before, the modified 
semidirect-product Lie-Poisson system \eqref{Casimir_dissipation_SDP} 
introduces dissipation of the Casimir $C$ while keeping energy conserved. This system recovers \eqref{SDP_LP_dissip} in the case when the bilinear form $ \gamma $ vanishes on $V$.
\end{remark} 

\paragraph{Lagrange-d'Alembert variational principle for semidirect products.} As in \eqref{EP_Casimir}, equations \eqref{SDP_LP_dissip} for Casimir dissipation in the semidirect product case can be expressed from the Lagrangian $\ell: \mathfrak{g}  \times V ^\ast \rightarrow \mathbb{R}  $ associated to $h$ via the Legendre transformation
\begin{equation}\label{Leg_transf_SDP} 
h( \mu,a)= \left\langle \mu , \xi \right\rangle - \ell( \xi,a ), \quad  \mu := \frac{\delta \ell}{\delta \xi}.
\end{equation} 
When $ \theta =0$, the constrained variational principle associated to the equations reads
\[
\delta \int_0 ^T  \ell( \xi ,a) dt=0, \quad\text{for variations} \quad\delta \xi = \partial _t \zeta -[ \xi , \zeta ], \;\;\; \delta a=- a \zeta ,
\]
where $ \zeta \in \mathfrak{g}  $ is an arbitrary curve vanishing at $t=0,T$, see \cite{HMR1998}.
As one may verify directly, when $ \theta \neq 0$, the dissipative Casimir force can be included in the Euler-Poincar\'e equations by considering the Lagrange-d'Alembert variational principle
\begin{equation}\label{LdA_SDP} 
\delta \left[ \int_0 ^T \ell( \xi ,a) dt\right] + \theta \int_0 ^T \gamma \left( \left[ \frac{\delta C}{\delta \mu } , \xi \right] , [ \xi ,\zeta] \right) dt  =0,
\end{equation} 
for variations $\delta \xi = \partial _t \zeta -[ \xi , \zeta ]$, $\delta a=- a \zeta$.

\subsection{Rotating shallow water (RSW) flows}\label{RSW-sec}

On the semidirect product of Diff$(\mathcal{D} )$ with functions $\eta \in\mathcal{F}(\mathcal{D} )$, where Diff$(\mathcal{D} )$ is the group of diffeomorphisms of a two dimensional domain $\mathcal{D} $, the Lie-Poisson equations \eqref{SDP_LP} become
\begin{equation}\label{compressible_LP} 
\partial _t \left( \frac{\mathbf{m} }{\eta }\right) + \operatorname{curl}  \left( \frac{\mathbf{m} }{\eta }\right) \times \frac{\delta h}{\delta \mathbf{m} }  + \nabla \left( \frac{\delta h}{\delta \mathbf{m} }  \cdot \frac{\mathbf{m} }{\eta }+ \frac{\delta h}{\delta \eta } \right) = 0,\quad \partial _t \eta + \operatorname{div} \left(\eta  \frac{\delta h}{\delta \mathbf{m} }  \right) =0.
\end{equation} 
The Hamiltonian for the rotating shallow water (RSW) equations is
\begin{equation}\label{Ham_SW} 
h( \mathbf{m} , \eta) = \int_ \mathcal{D} \left( 
\frac{1}{2\eta}| \mathbf{m}-\eta \mathbf{R}  | ^2 + \frac{1}{2} g(\eta-\bar{D}) ^2\right) 
\,dxdy\,,
\end{equation} 
where $ \operatorname{curl} \mathbf{R} = 2 \boldsymbol{\Omega}  $ is the Coriolis parameter, $\eta$ is the total depth of the bottom topography in the two dimensional domain with spatial coordinates $(x,y)$ and $\bar{D}(x,y)$ is the mean depth.   
Indeed, inserting this Hamiltonian into the Lie-Poisson equation \eqref{compressible_LP} yields the rotating shallow water equations for the fluid velocity $\mathbf{u}=\mathbf{m}/\eta-\mathbf{R}$,
\begin{equation}\label{motion_SW} 
\partial _t \mathbf{u} + \operatorname{curl} (\mathbf{u} + \mathbf{R} )\times \mathbf{u} + \nabla \left(\frac{1}{2} | \mathbf{u} | ^2 + g(\eta-\bar{D}) \right) =0
, \qquad 
\partial _t \eta + \operatorname{div}( \eta \mathbf{u} )=0. 
\end{equation} 
The second term may be written as $ \eta^{-1}\operatorname{curl} (\mathbf{u} + \mathbf{R} ) \times \eta\mathbf{u}= q \mathbf{\hat{z}} \times \eta \mathbf{u} $, with the potential vorticity (PV) $q$ given by
\[
q 
= \eta^{-1} \mathbf{\hat{z}}\cdot \operatorname{curl} (\mathbf{m}/\eta)
= \eta^{-1} \mathbf{\hat{z}}\cdot \operatorname{curl} (\mathbf{u} + \mathbf{R} )
\]
As a consequence of the shallow water equations, the potential vorticity $q$ is advected (conserved on fluid parcels), that is
\[
\partial _tq  + \mathbf{u}\cdot \nabla q =0
\,.
\]
Of course, the conservation of PV is not limited to the shallow water equations. In fact, it holds for the equations derived from this LP structure for \emph{any} Hamiltonian. This is because the  LP structure admits PV in a family of Casimir functions that Poisson commute with every functional of the variables $(\mathbf{m},\eta)$.

\paragraph{Casimir functions for RSW.}
One class of Casimir functions for RSW is given by
\[
C _{\Phi } ( \mathbf{m} , \eta ) 
= \int_ \mathcal{D} \eta \,\Phi \left(q  \right)dxdy
 , \quad\text{where} \quad   
q = \eta^{-1} \mathbf{\hat{z}} \cdot \operatorname{curl} \left(\mathbf{m}/\eta\right),
\]
for any smooth function $\Phi$. We will compute the Casimir dissipative version of RSW associated to the Casimir function
\begin{equation}\label{Casimir_RSW} 
C( \mathbf{m} , \eta )=\frac{1}{2}  \int_ \mathcal{D} \eta  \,q^2dxdy \,,
\end{equation} 
and the inner product 
\begin{equation}\label{inner_eta}
\gamma_{\eta } ( \mathbf{u} , \mathbf{v} )=\int_ \mathcal{D} \eta \,( \mathbf{u} \cdot \mathbf{v} )dx dy
\,.
\end{equation}
The flat operator associated to this inner product is thus $ \mathbf{v} ^\flat = \eta \mathbf{v}\cdot d \mathbf{x}  $.
From the abstract Lie algebraic formulation \eqref{SDP_LP_dissip} we obtain
\begin{equation}\label{modifiedSW} 
\partial _t \mathbf{u} + \operatorname{curl}( \mathbf{u} +\mathbf{R} 
+ \mathbf{A}) \times \mathbf{u} 
+ \nabla \left( 
\frac{1}{2} | \mathbf{u} | ^2 + g( \eta -\bar D) 
+
\mathbf{u} \cdot \mathbf{A}\right) 
=0
\,,
\end{equation}
where the vector $\mathbf{A}$ is defined as
\begin{equation}\label{Avec-def}
\mathbf{A}:=\theta [\mathbf{u} ,\mathbf{\hat{z}} \times (\nabla q)/\eta  ]
\,.
\end{equation}

Equations (\ref{modifiedSW}) and (\ref{Avec-def}) are obtained from the Lie-Poisson equations \eqref{compressible_LP}, where in the second and third terms the expressions ${\mathbf{m}}/{ \eta } $ have been replaced by ${\tilde{ \mathbf{m} }}/{ \eta }$, where
\[
\mathbf{\widetilde{m}}
= \mathbf{m} -\theta \left[ \frac{\delta h}{\delta \mathbf{m} }, \frac{\delta C}{\delta \mathbf{m} } \right] ^\flat
= \mathbf{m} + \mathbf{A}^\flat =\mathbf{m} + \eta\mathbf{A},
\]
as dictated by the Lie-Poisson Casimir dissipation approach in formula \eqref{SDP_LP_dissip}. We have used the  formula
$\delta C/ \delta \mathbf{m} =  \operatorname{curl} \left(q \mathbf{\hat{z}} \right) / \eta  = -\mathbf{\hat{z}} \times (\nabla q)/ \eta$ for the variational derivative of the Casimir \eqref{Casimir_RSW} and have kept in mind that the Lie algebra bracket appearing in \eqref{SDP_LP_dissip} is minus the Lie bracket of vector fields. 
The last term of \eqref{modifiedSW} is computed as follows
\begin{align*} 
 \frac{\delta h}{\delta \mathbf{m} }  \cdot \frac{\tilde{\mathbf{m} }}{\eta}+ \frac{\delta h}{\delta \eta}& = \mathbf{u} \cdot \left( \mathbf{u}+ \mathbf{R}  + \theta \left[\mathbf{u} ,\mathbf{\hat{z}} \times (\nabla q )/ \eta \right] \right) - \frac{1}{2} |\mathbf{u} | ^2 - \mathbf{u} \cdot \mathbf{R} +g( \eta -\bar D)\\
 &= \frac{1}{2} | \mathbf{u} | ^2 + g( \eta -\bar D) 
+
\theta \mathbf{u} \cdot [\mathbf{u} ,\mathbf{\hat{z}} \times (\nabla q)/\eta]
\\&=
 \frac{1}{2} | \mathbf{u} | ^2 + g( \eta -\bar D) 
+
\mathbf{u} \cdot \mathbf{A}
,
\end{align*}
where we have used the definition of $\mathbf{A}$ in equation \eqref{Avec-def} and the relation $\mathbf{u}=\mathbf{m}/\eta-\mathbf{R}$ between the velocity and momentum in the rotating case.

A short calculation using equation \eqref{modifiedSW} yields
\begin{equation}\label{modif_PV}
\frac{Dq}{Dt} = \eta^{-1}\mathbf{\hat{z}} \cdot \operatorname{curl}
 \left(\mathbf{u} \times 
  \operatorname{curl} \mathbf{A} \right)
\quad\hbox{with}\quad
\frac{D}{Dt} := \partial _t + \mathbf{u}\cdot\nabla
\,. \end{equation}
This model preserves the energy in \eqref{Ham_SW} and dissipates the enstrophy Casimir $ \frac{1}{2} \int_ \mathcal{D}\eta\, q ^2 dxdy$ as
\begin{align*}
\frac{d}{dt} \frac{1}{2} \int_ \mathcal{D} \eta \, q ^2 dxdy
= -\, \theta \int_ \mathcal{D} \eta \big|[\mathbf{\hat{z}} \times (\nabla q)/\eta \,,\, \mathbf{u} ]\big|^2 dxdy
= -\, \theta^{-1} \int_ \mathcal{D} \eta |\mathbf{A}|^2 dxdy\,.
\end{align*}

\paragraph{Kelvin circulation theorem.} A convenient way to derive the Kelvin circulation theorem \eqref{KN_SDP} is to rewrite the Lie-Poisson equations \eqref{compressible_LP} using the duality pairing with one-forms. We get
\begin{equation}\label{LP_RSW_oneform} 
\partial _t \left( \frac{\alpha }{\eta }\right)   
+ \pounds _{ \frac{\delta h}{\delta \alpha }} \left( \frac{\alpha }{\eta }\right) 
+ d\left(\frac{\delta h}{\delta \eta }\right)
=0  
\,, \qquad 
\partial _t \eta +\operatorname{div} \left(\eta  \frac{\delta h}{\delta\alpha }  \right )=0
\,,\end{equation} 
in which  $\alpha$ denotes the 1-form $\alpha=\mathbf{m} \cdot d\mathbf{x}$
and $\mathcal{L}_{\delta h/\delta \alpha}$ denotes the Lie derivative with respect to the vector field $\delta h/\delta \alpha$, as defined for 3D earlier in \eqref{Lie-der3D-def}.
The Casimir dissipative LP equation is obtained by replacing $ \alpha $ by $\tilde{ \alpha }= \tilde{ \mathbf{m} }\cdot d \mathbf{x} $, in  the second term. 
Consequently, a direct computation yields
\begin{equation}\label{gen2D_Kelcircthm} 
\frac{d}{dt} \oint_{c(u)} \frac{\alpha }{\eta }=  \theta \oint_{c(u)} \mathcal{L} _{ \frac{\delta h}{\delta \alpha }}  \left( \frac{1 }{\eta }\left[ \frac{\delta h}{\delta \alpha }, \frac{\delta C}{\delta \alpha }\right] ^\flat\right) 
=\theta \oint_{c(u)} \mathcal{L} _{ \frac{\delta h}{\delta \alpha }}\left(  \left[ \frac{\delta h}{\delta \alpha }, \frac{\delta C}{\delta \alpha }\right]\cdot d \mathbf{x} \right)  ,
\end{equation}
consistently with \eqref{KN_SDP}, where we recall that the Lie algebraic bracket is minus the Lie bracket of vector fields. For RSW, the previous equation produces the following Kelvin circulation theorem
\begin{equation}\label{modRSW_Kelcircthm} 
\frac{d}{dt} \oint_{c(u)}( \mathbf{u} + \mathbf{R} ) \cdot d\mathbf{x} 
 = - \oint_{c(u)} \mathcal{L}  _ u
 \left( \mathbf{A} \cdot d \mathbf{x} \right)
 =
  \oint_{c(u)} \left(\mathbf{u} \times 
  \operatorname{curl} \mathbf{A} \right) \cdot d \mathbf{x}   
  ,
\end{equation}
with vector $\mathbf{A}$ defined in equation \eqref{Avec-def}.

\paragraph{Modified PV advection.} One way to obtain the modified advection equation \eqref{modif_PV} for PV, is to take the exterior differential of the first equation in \eqref{LP_RSW_oneform} (with $ \alpha $ replaced by $ \tilde{ \alpha }$ in the second term) and obtain
\[
\partial _t d\left( \frac{\alpha }{\eta } \right) 
+ \mathcal{L} _u d\left( \frac{\alpha }{\eta } \right)
= -\, \mathcal{L} _u d(\mathbf{A}\cdot d \mathbf{x} ).
\]
From this, and from the second equation in \eqref{LP_RSW_oneform} one obtains
\[
\partial _t q + \mathbf{u} \cdot \nabla q
=  \frac{1}{\eta }  \mathbf{\hat{z}} \cdot 
\operatorname{curl} \,(
\mathbf{u}\times \operatorname{curl}\,
\mathbf{A}),
\]
by making use of the formulas
\[
d\left( \frac{\alpha }{\eta }\right) 
= q\,\eta\, dxdy
\quad \text{and} \quad 
d( \mathbf{m}  \cdot d\mathbf{x} )= (\mathbf{\hat{z}} \cdot \operatorname{curl} \mathbf{m} )\, dxdy
\]
for one-forms $ \alpha $ and vector fields $ \mathbf{m} $ on the two dimensional domain $ \mathcal{D} $.

\paragraph{Lagrange-d'Alembert variational principle.} From the Legendre transform \eqref{Leg_transf_SDP}, one obtains the RSW Lagrangian
\[
\ell( \mathbf{u} , \eta )= \int_ \mathcal{D}\left[  \eta \left(\frac{1}{2} | \mathbf{u} | ^2 + \mathbf{u} \cdot \mathbf{R}  \right) - \frac{1}{2} g( \eta -\bar D) \right] dxdy
\,.
\]
From \eqref{LdA_SDP} we get the Lagrange-d'Alembert variational principle,
\[
\delta \int_0^T \ell( \mathbf{u} , \eta ) dt+ \theta \int_0^T\int_ \mathcal{D} \eta ([ \mathbf{u} , \mathbf{\hat{z}} \times (\nabla q)/ \eta ] \cdot [ \mathbf{u} , \mathbf{v} ] )dxdy
=0\,,
\]
This variational principle yields the Casimir dissipation constraint force in RSW, i.e., the terms proportional to $\theta$ in the modified RSW motion equation \eqref{modifiedSW},
for constrained variations $\delta \mathbf{u}  = \partial _t \mathbf{v}  +[ \mathbf{u}  , \mathbf{v}  ]$, $\delta \eta =- \operatorname{div}( \eta \mathbf{v} )$, where $ \mathbf{v} $ is and arbitrary time dependent vector field vanishing at $t=0, T$, associated to a virtual displacement $ \delta \varphi $ of the fluid configuration $ \varphi $.

\subsection{3D rotating Boussinesq flows}

\paragraph{Boussinesq equations for rotating stratified incompressible 3D flows.}
On the semidirect product of SDiff$(\mathbb{R}^3)$ with functions $b\in\mathcal{F}(\mathbb{R}^3)$, where SDiff$(\mathbb{R}^3)$ is the group of volume preserving diffeomorphisms of $ \mathbb{R}  ^3 $, the Lie-Poisson equations \eqref{SDP_LP} become
\begin{equation}\label{SDP_Boussinesq} 
\partial _t \mathbf{q} 
+ \operatorname{curl} \left( \mathbf{q}  \times \operatorname{curl}\frac{\delta h}{\delta \mathbf{q}  }\right) - \nabla \frac{\delta h}{\delta b }\times  \nabla b =0, \qquad 
\displaystyle\partial _t b + \operatorname{curl}\frac{\delta h}{\delta \mathbf{q}  }\cdot \nabla b  =0.
\end{equation}
For rotating 3D Boussinesq fluids, the Hamiltonian is
\begin{equation}\label{erg_Boussinesq} 
h( \mathbf{q} , b)= \int\left(  \frac{1}{2} | \operatorname{curl} ^{-1} \mathbf{q}  - \mathbf{R} | ^2 + bz \right)  d ^3 x,
\end{equation}
where $b$ is the scalar buoyancy and $ \operatorname{curl} \mathbf{R} = 2 \Omega $ is the Coriolis parameter. The variational derivatives of the Hamiltonian $h( \mathbf{q} , b)$ are
\[
\operatorname{curl} \frac{\delta h}{\delta \mathbf{q}  } = \operatorname{curl} ^{-1} \mathbf{q}  - \mathbf{R}   = \mathbf{u} , \qquad \frac{\delta h}{\delta b} = z
\,,\]
and the vorticity is $ \boldsymbol{\omega} = \operatorname{curl} \mathbf{u} =\mathbf{q} - \operatorname{curl} \mathbf{R} = \mathbf{q} - 2 \boldsymbol{\Omega}$.
In this case, the Lie-Poisson equations \eqref{SDP_Boussinesq} reduce to the rotating 3D Boussinesq equations for vorticity
\begin{equation}\label{rotBoussinesq-vort} 
\partial _t \boldsymbol{\omega} + \operatorname{curl} \left( \left( \boldsymbol{\omega} + 2 \boldsymbol{\Omega} \right) \times \mathbf{u} \right)  - \mathbf{\hat{z}} \times \nabla b=0, \qquad \partial _t b + \mathbf{u} \cdot \nabla b  =0
\,,
\end{equation}
or, in velocity form
\begin{equation}\label{rotBoussinesq-vel} 
\partial _t \mathbf{u} + (\operatorname{curl} \mathbf{u} + 2 \boldsymbol{\Omega} ) \times \mathbf{u}  + \mathbf{\hat{z}} b=- \nabla p,
\end{equation}
where $p$ is the pressure, found by enforcing incompressibility,  ${\rm div}\,\mathbf{u}=0$.

\begin{remark}[Rotating 3D Boussinesq Casimir functions]\rm
The Lie-Poisson equations \eqref{SDP_Boussinesq} admit the Casimir functions
\[
C_{ \Phi }(\mathbf{q} , b )
=\int _ \mathcal{D} \Phi (q, b ) \,d ^3 x
, \quad \hbox{where} \quad
q
:= \nabla b \cdot \mathbf{q} 
\quad \hbox{is the potential vorticity.} 
\]
\end{remark}

\paragraph{Casimir dissipation.}
We shall now introduce the dissipative Casimir effect for the 3D rotating Boussinesq system \eqref{SDP_Boussinesq} by using the Lie algebraic setting developed above. Choosing the \emph{enstrophy} Casimir for potential vorticity 
\begin{equation}\label{Casimir_3D} 
C(\mathbf{q}, b )=\frac12\int_ \mathcal{D} q ^2  d^3x,
\end{equation} 
we have
\[
\operatorname{curl} \frac{\delta C}{\delta \mathbf{q} }
= \operatorname{curl}( q \nabla b) , \quad \frac{\delta C}{\delta b } 
= - \operatorname{div}(q \mathbf{q}  )
\,.\]
Consequently, choosing the inner product $ \gamma $ associated to a given positive self-adjoint differential operator $ \Lambda $, the modified momentum \eqref{Casimir_dissipation_tilde} is
\[
\mathbf{\widetilde{q}}
= \mathbf{q} - \theta \operatorname{curl} \Lambda \left[\mathbf{u},\operatorname{curl}( q \nabla b)\right]
= \mathbf{q} - \theta L(\operatorname{curl}( q \nabla b)\times \mathbf{u})
:= \mathbf{q} + \operatorname{curl}\mathbf{A}, 
\]
where the differential operator $\Lambda$ verifies $ \Lambda = {\rm curl} ^{-1} \,L\ {\rm curl} ^{-1} $ and $L = {\rm curl}\,\Lambda\,{\rm curl}$, so $\operatorname{div}\mathbf{\widetilde{q}}=0$, with 
\begin{equation}\label{curlA-def} 
\operatorname{curl}\mathbf{A}
:=
\theta L( \mathbf{u} \times \operatorname{curl}( q \nabla b))
=
\mathbf{\widetilde{q}} - \mathbf{q}
\,.
\end{equation} 

\paragraph{Potential vorticity.}
The dissipative Casimir Lie-Poisson system 
\begin{equation}\label{DissipCasimir_3D} 
\partial _t \mathbf{q} 
+ \operatorname{curl} \left(\mathbf{\widetilde{q}}  \times \operatorname{curl}\frac{\delta h}{\delta \mathbf{q}  }\right) - \nabla \frac{\delta h}{\delta b }\times  \nabla b =0
, \quad \partial _t b+ \mathbf{u} \cdot \nabla b=0
\,,
\end{equation} 
implies the following equation for the potential vorticity, given by the projection $q=\mathbf{q}\cdot \nabla{b}$ of the total vorticity $\mathbf{q}$ onto a vector normal to a level set of buoyancy $b$,
\begin{equation}\label{BoussPV_3D} 
\frac{D}{Dt}(\mathbf{q}\cdot \nabla{b})
+ 
\mathbf{q}\cdot \nabla \frac{Db}{Dt}
=
\nabla{b} \cdot
\operatorname{curl} \left(\mathbf{u} \times 
  \operatorname{curl} \mathbf{A} \right)
\quad\hbox{with}\quad
\frac{Db}{Dt} = 0
\,.
\end{equation} 
This means that the dissipative Casimir system \eqref{DissipCasimir_3D} \emph{does not conserve} the standard Boussinesq potential vorticity along Lagrangian fluid parcels. That is, ${Dq}/{Dt} \ne 0$ for this system. This follows because the Ertel theorem no longer holds for  fluid flows with  Casimir dissipation. 

\paragraph{Energy and enstrophy.} 
In velocity form, the system \eqref{DissipCasimir_3D} is given by
\begin{equation}\label{BoussVort_3D} 
\partial _t \mathbf{u} 
-  \mathbf{u} \times
\operatorname{curl}\left( \mathbf{u}+ \mathbf{R}  + \mathbf{A} \right) 
+ b\,\mathbf{\hat{z}}
=- \nabla p
, \quad \partial _t b+ \mathbf{u} \cdot \nabla b=0.  
\end{equation} 
The first of these two equations can be rewritten, in terms of $L$ as
\begin{equation}\label{3D_velocity_form}
\partial _t \mathbf{u} 
-  \mathbf{u} \times
\operatorname{curl}\left( \mathbf{u}+ \mathbf{R}   \right) 
-  \mathbf{u} \times \big(
\theta L( \mathbf{u} \times \operatorname{curl}( q \nabla b))\big)
+ b\,\mathbf{\hat{z}}
=- \nabla p.
\end{equation} 
This system preserves the energy $h( \mathbf{q} , b)$ of the 3D rotating Boussinesq fluid in \eqref{erg_Boussinesq} and dissipates the enstrophy as
\begin{equation}\label{KN_3D_Boussinesq} 
\frac{d}{dt} \frac{1}{2} \int_ \mathcal{D} q ^2 d ^3 x= -\, \theta\left  \|\big[ \operatorname{curl}(q \nabla b ) ,\, \mathbf{u} \big]\right \| ^2 _ \Lambda =  -\, \theta\left  \|\operatorname{curl}(q \nabla b ) \times \mathbf{u}\right \| ^2 _ L
\,,
\end{equation} 
where we have used the identity (\ref{LLnorms-id}).

\paragraph{Kelvin circulation theorem.} A convenient way to derive the Kelvin circulation theorem \eqref{KN_SDP} for rotating Boussinesq flows with Casimir dissipation is to rewrite the modification of equations \eqref{SDP_Boussinesq} using duality pairing with one-forms. In this case, we have
\[
\partial _t \alpha + \mathcal{L} _{ \frac{\delta h}{\delta \alpha }} \tilde { \alpha } - \frac{\delta h}{\delta b} db=-d p ,
\]
in which  $\alpha$ denotes the 1-form $\alpha=\balpha\cdot d\mathbf{x}$ or, more explicitly,
\[
\partial _t \alpha + \mathcal{L} _{ \frac{\delta h}{\delta \alpha }} \alpha = \theta \mathcal{L} _{ \frac{\delta h}{\delta \alpha }} \left[ \frac{\delta h}{\delta \alpha }, \frac{\delta C}{\delta \alpha }\right] ^\flat    +\frac{\delta h}{\delta b} d b-dp
\,.\]
Consequently, a direct computation yields
\[
\frac{d}{dt} \oint_{c(u)} \alpha =  \theta \oint_{c(u)} \mathcal{L} _{ \frac{\delta h}{\delta \alpha }} \left[ \frac{\delta h}{\delta \alpha }, \frac{\delta C}{\delta \alpha }\right] ^\flat -\oint_{c(u)} b\,d \frac{\delta h}{\delta b}
\,,
\]
consistently with \eqref{KN_SDP}, where we again recall that the Lie algebraic bracket is minus the Lie bracket of vector fields.
For rotating 3D Boussinesq, upon using the enstrophy as a Casimir function and employing the differential operator $ \Lambda $, one finds
\begin{align*} 
\frac{d}{dt} \oint_{c(u)}( \mathbf{u} + \mathbf{R} ) \cdot d\mathbf{x}  
&=  \theta \oint_{c(u)} \mathcal{L} _{\mathbf{u} } \left(\Lambda  \left[ \mathbf{u} , \operatorname{curl}(q \nabla b ) \right]\cdot d \mathbf{x} \right)   
-\oint_{c(u)} b\, dz
\\&=
-\,\theta \oint_{c(u)} \mathbf{u} \times L\left(\operatorname{curl}(q \nabla b ) \times \mathbf{u} \right) \cdot d \mathbf{x}   -\oint_{c(u)} b\, dz,
\\&=
 \oint_{c(u)} \left(\mathbf{u} \times 
 \operatorname{curl} \mathbf{A} \right) \cdot d \mathbf{x}   
 -\oint_{c(u)} b\, dz,
\end{align*} 
Setting $ \theta =0$ recovers the usual Kelvin circulation theorem for rotating Boussinesq flows.

\noindent
Here we have used the expression \eqref{erg_Boussinesq} for the Hamiltonian in terms of $ \balpha $ and $b$ to compute variational derivatives, that is,
\[
h( \alpha , b)
= \int_ \mathcal{D}\left(  \frac{1}{2} | \balpha - \mathbf{R} | ^2 + bz \right) d^3x
, \quad 
\frac{\delta h}{\delta \balpha }= \mathbf{u}
 , \quad\text{with} \quad   
 \balpha = \mathbf{u} + \mathbf{R} 
 \,.
\]

\paragraph{Lagrange-d'Alembert variational principle for rotating Boussinesq flows.} Proceeding exactly as in the preceding example for RSW and using the abstract formulation \eqref{LdA_SDP}, one obtains the Lagrange-d'Alembert variational principle  for Boussinesq flows,
\[
\delta \int_0^T \ell( \mathbf{u} , b ) dt+ \theta \int_0^T\int_ \mathcal{D} ( \operatorname{curl}(q \nabla b) \times  \mathbf{u} ) \cdot L( \mathbf{u} \times \mathbf{v} ) dxdy = 0,
\]
for variations $\delta \mathbf{u}  = \partial _t \mathbf{v}  +[ \mathbf{u}  , \mathbf{v}  ]$, $\delta b =- \mathbf{v} \cdot \nabla b$. Here $ \mathbf{v} $ is an arbitrary time dependent vector field vanishing at $t=0, T$, that is associated to a virtual displacement $ \delta \varphi $ of the fluid configuration $ \varphi \in \operatorname{SDiff}( \mathbb{R}  ^3 )$. Also, $\ell$ is the Lagrangian of the 3D rotating Boussinesq fluid, given by
\[
\ell( \mathbf{u} , b)= \int_ \mathcal{D} \left( \frac{1}{2} | \mathbf{u} | ^2 + \mathbf{u} \cdot \mathbf{R} - bz\right) d ^3 x.
\]

\begin{remark}[Impermeability theorem]\rm
In the Casimir dissipation approach for rotating 3D Boussinesq fluids, the vector $\operatorname{curl}(q \nabla b ) = \nabla q\times\nabla b$ in equations \eqref{3D_velocity_form} and \eqref{KN_3D_Boussinesq} plays the same role that the vorticity $ \boldsymbol{\omega} $ played for ideal incompressible 3D in equations \eqref{3D_Casimir_velocity} and \eqref{Helicitysquared-dissip}. The vector $\nabla q\times\nabla b$ also plays a role in the \emph{impermeability theorem} of \cite{HMc90}, which applies, for example, in the dynamics of intersections of level sets of potential temperature and potential vorticity in the atmospheric tropopause. See also, \cite{GiHo2010,GiHo2011} for an additional mathematical treatment of rotating 3D Boussinesq dynamics.

In the Casimir dissipation approach, the impermeability theorem no longer holds. Instead, one finds from (\ref{BoussPV_3D}) that 
\begin{equation}\label{Bouss-dq^db} 
\partial_t(q \nabla b) - \mathbf{u} \times \operatorname{curl}(q \nabla b)
+ \nabla (q\mathbf{u}\cdot\nabla b)
=
\left(\frac{Dq}{Dt}\right)\nabla b
=
\big(\nabla{b} \cdot
\operatorname{curl} \left(\mathbf{u} \times 
  \operatorname{curl} \mathbf{A} \right)\big)\nabla b
\,.
\end{equation} 
Taking the curl of this equation shows that intersections of level sets of $q$ and $b$ are not frozen into the flow, which contradicts the departure point for the impermeability theorem  of \cite{HMc90}.

\end{remark}


\section{Conclusion} 
\label{conclusion-sec}

This paper has introduced a theory of selective decay by Casimir dissipation which applies widely in fluids and suggests a connection between selective decay and the design of numerical methods, since it recovers some of the previous methods, such as the anticipated vorticity method.

Interestingly, while the enstrophy $C_2$ in (\ref{enstrophy-def}) is the appropriate Casimir for selective decay in 2D incompressible flow, the appropriate Casimir for that role in 3D incompressible flow is the \emph{square} of the helicity $C$ in (\ref{helicity-def}).

Perhaps one reason for these associations with different Casimirs in 2D and 3D can be suggested by comparing the Lie-Poisson brackets in terms of the vorticity in 2D and 3D. In particular, these brackets may be written in terms of the Casimirs, the enstrophy $C_2$ in 2D and the helicity $C$ in 3D. In fact, the Lie-Poisson brackets in 2D and 3D may both be written in terms of the corresponding  Casimirs in the same triple product form \cite{Nambu1973}. Namely,
\[
\{f,\,h\}_2 (\omega)
= \int_\mathcal{D}\frac{\delta C_2}{\delta \omega}  \left[ \frac{\delta f}{\delta \omega } , \frac{\delta h}{\delta \omega } \right] 
\,dxdy
\quad\hbox{and}\quad
\{f,\,h\}_3 (\bom)
= \int_\mathcal{D}
{\rm curl}\,\frac{\delta C}{\delta \bom }
\cdot  {\rm curl}\,\frac{\delta f}{\delta \bom } 
\times {\rm curl}\,\frac{\delta h}{\delta \bom }
\,d^3x
\,.
\]
These formulas suggest the fundamental roles of the two different Casimirs, the enstrophy $C_2$ in 2D for which $\delta C_2/\delta \omega = \omega$ and the helicity $C$ in 3D for which $\delta C/\delta \bom = \mathbf{u}$. One could imagine using other constants of motion instead of the Casimirs, but these would result by Noether's theorem from particular symmetries and their use would be restricted to those subcases. 

The loss of Kelvin's theorem, as well as the losses of potential vorticity conservation and the Casimir conservation laws as a result of modifying the equations to impose selective decay all would be consistent with the effects of viscosity, but of course the preservation of energy is not. This, however, is the essence of selective decay, which is seen in turbulence on times scales less than the time scale for viscous decay of energy.

The phenomenon of selective decay is thought to arise from a multiscale  interaction between disparate large and small scales in which the large scales are regarded as a type of coarse graining, or perhaps the envelope, of the small scale motions,  \cite{MiPoSu2008}. This type of model is consistent with other recently developed multiscale turbulence models, particularly that in \cite{HoTr2012}, in which the large scale motion is regarded as a \emph{Lagrange coordinate} for the small scale motion. 

Thus, a phenomenological point of view exists in which the effect of selective decay of Casimirs that depend on gradients of the velocity may be regarded usefully as providing a type of coarse-graining that is similar to the effect of a turbulent viscosity, but still conserves energy. In this regard, see the work of  \cite{VaHu1988,GrSa1989,LeDuCh2006} that investigates the emergence of large-scale flows under weak viscous decay.

This interpretation of selective decay arising from coarse graining may also suggest a connection between selective decay and the design of numerical methods. In particular, one may wish to design numerical methods whose solutions imitate selective decay without adding too much linear dissipation of energy and momentum due to viscosity. And from our discussion here it seems that the method of anticipated vorticity in \cite{SaBa1981,SaBa1985} may be one of those numerical methods. 

The theoretical approach presented here may lead to classes of other numerical methods based on selective Casimir decay that would generalize and extend the applicability of the method of anticipated vorticity to other areas of fluid dynamics. A recent discussion of numerical methods based on anticipated vorticity for 2D flows is given in \cite{ChGuRi2011}. A ``scale-aware'' version of the anticipated vorticity method in 3D should also be readily obtainable, by allowing the value of its time-scale parameter $\theta$ to depend on local mean properties of the solution. However, this is beyond the scope of the present work.

If selective Casimir decay or the anticipated vorticity method does mimic the effect of linear Navier-Stokes viscosity, but at a convenient time scale for numerical simulations, then one must be cautious about adding in other types of viscosity, such as the artificial viscosity used in Large Eddy Simulations to mimic the effects of turbulence. 
The introduction of other types of turbulent viscosities such as Smagorinsky tensor diffusivity in \emph{addition} to anticipated vorticity for selective decay could  be seen as double counting. 

In contrast, the present selective decay approach is complementary to the introduction of nondissipative regularization through nonlinear dispersion as a modification of the kinetic energy Hamiltonian, as accomplished in the Navier-Stokes alpha (NS-alpha) turbulence model \cite{FoHoTi2001,FoHoTi2002} or its alternative form in \cite{Ve2008}, both of which conserve energy, enstrophy (in 2D) and helicity. For example, one may easily apply the Casimir-dissipation approach to the Euler-alpha model, because the latter is Lie-Poisson in the same way as the Euler fluid equations. Moreover, the original derivation of the NS-alpha turbulence model by the method of Lagrangian averaging makes it clear that the model is already  ``scale-aware'' because the length scale alpha in that derivation is the correlation length for Lagrangian fluctuations of the fluid parcel paths. The NS-alpha model can also be made \emph{actively} scale aware by allowing the length scale alpha to evolve with the flow. 

In regard to other modifications of the fluid nonlinearity, the Casimir-dissipation approach is also complementary to the \cite{CrLe1976} approach, which introduces a vortex force induced by Stokes drift. See \cite{Ho1996} for an analysis of the Craik--Leibovich approach from the Lie-Poisson viewpoint. The Casimir-dissipation approach also differs from the theory of bolus transport by fluctuations of \cite{GeMcW1996}, in that only one transport velocity is present in the Casimir-dissipation approach. 

\paragraph{Outlook.}
Future applications of our approach here would naturally follow recent work in ocean dynamics that interprets selective decay as a mechanism for parameterizing the interactions between disparate scales, for example, between large coherent oceanic flows and the much smaller eddies directly, as done in \cite{MaAd2010}, instead of relying on the slower, indirect effects of viscosity. 
The Lie-Poisson structures in the selective decay models presented here suggest a framework in which numerical schemes for their computational simulations may be designed. This framework may turn out to be similar to the ideas in \cite{Sa2005} about using the triple product Nambu form to design simulation algorithms for certain grid structures used in oceanography. In addition, Voronoi mesh methods that were developed for anticipated vorticity dynamics may turn out to be naturally adaptable to the the selective decay models presented here. See. e.g., \cite{ChGuRi2011,ChGuRi2012}. However, applications of these ideas for numerical simulation algorithms within the Casimir dissipation modelling framework are beyond our present scope and will be pursued elsewhere.

Analytical issues such as whether the 3D selective decay equations admit existence of unique globally strong solutions were not addressed here. 
These analytical issues turn on answering the question of whether the modified vortex stretching term in equations \eqref{3Dvort-Leqn} and \eqref{vortstretchcirc} either enhances or diminishes circulation. This is an outstanding mathematical issue that we have been unable to address in this paper. As for other aspects of selective decay, perhaps a statistical mechanics approach would be useful here. 

In combining the NS-alpha model with the present selective decay approach, one may expect existence and uniqueness to persist, since this property is already possessed by the NS-alpha turbulence model, which controls the vortex stretching term for the Navier-Stokes equations. However, the vector $\mathbf{B}$ defined in \eqref{Bvector-def} that appears in the coefficient $[\mathbf{u},\,\mathbf{B}]$ in the modified vortex stretching term appearing in equations \eqref{3Dvort-Leqn} and \eqref{vortstretchcirc} depends on higher-order derivatives associated with the operator $L$. Consequently, the expectation that the NS-alpha should control the modified vortex stretching term may not be fulfilled for the standard $H^1$ regularisation. Of course, the higher Sobolev norm regularisation norms $\|\,\cdot\,\|^2_L$ are also available for use in the NS-alpha formulation. These higher Sobolev norms in combination with the NS-alpha formulation of \cite{FoHoTi2001,FoHoTi2002} may be required to guarantee analytical control of the vortex stretching term for the Casimir-dissipation selective decay models introduced here. 

The Casimir dissipation approach for other fluid models such as magnetohydrodynamics (MHD), 
compressible fluids, two-layer quasi-geostrophic models parameterizing baroclinic instability, 
Maxwell-fluid and Maxwell-Vlasov plasmas, spin chains, solitons, complex fluids, turbulence models, etc., will be investigated elsewhere. 

\subsection*{Acknowledgments} We are grateful to R. Hide, C. Tronci, G. Vallis and B. A. Wingate for fruitful and thoughtful discussions during the course of this work. FGB was partially supported by a ``Projet Incitatif de Recherche'' contract from the Ecole Normale Sup\'erieure de Paris. DDH is grateful for partial support by the European Research Council Advanced Grant 267382 FCCA.

{\footnotesize

\bibliographystyle{new}

\end{document}